\documentclass[12pt,final]{article}
\usepackage{amssymb,paralist}
\usepackage[raiselinks=false,colorlinks=true,citecolor=blue,
    urlcolor=blue,linkcolor=blue,bookmarksopen=true]{hyperref}
\usepackage{calc,dashbox}
\usepackage{bookmark}
\usepackage{fullpage}
\usepackage[normalem]{ulem}
\newtheorem{theorem}{Theorem}
\newtheorem{lemma}[theorem]{Lemma} 

\makeatletter
 \newcommand\setreflabel[1]{\protected@edef\@currentlabel{#1}}\newcounter{claim}
 
 \newcommand\claimlabelfmt[1]{(#1)}
 \newenvironment{claim}[1][]%
   {\removelastskip\ifx\newenvironment#1\newenvironment%
    \refstepcounter{claim}\def\@claim{\arabic{claim}}\else\def\@claim{#1}\fi%
    \setreflabel{\expandafter\claimlabelfmt{\@claim}}%
    \global\edef\@lastclaim{clm@\the\inputlineno}%
    \label{\@lastclaim}\par\vspace{1ex}%
    \begin{compactitem}[~{\@currentlabel}]\item\it}
   {\end{compactitem}\par\vspace{1ex}}
 \newenvironment{proofclaim}[1][]%
   {\par\edef\@thisclaim{\ifx\newenvironment#1\newenvironment%
    \@lastclaim\else#1\fi}\noindent\ignorespaces}
   {This proves~\expandafter\ref{\@thisclaim}.\par\vskip 1ex\relax}
 \newenvironment{innerclaim}%
   {\let\@oldclaim\@thisclaim}
   {\let\@thisclaim\@oldclaim}
 \makeatother
\newcommand{\qed}{$\Box$}
\newenvironment{proof}%
  {\noindent{\bf Proof.}\ }%
  {\hfill\qed\par\bigskip}
  {\noindent{\bf Proof of #1.}\ }%
  {\hfill\qed\par\bigskip}
\newcommand{\tp}{\mbox{\hspace*{0.1em}--\hspace*{0.1em}}}
\begin{document}

\title{4-coloring $P_6$-free graphs with no induced 5-cycles\thanks{Authors partially supported by NSF grants IIS-1117631 and DMS-1265803.}
\thanks{A short version of this paper has been submitted to the ACM-SIAM Symposium on
Discrete Algorithms.}
}
\author{Maria Chudnovsky, Peter Maceli, Juraj Stacho, Mingxian Zhong}
\date{\small IEOR Department, Columbia University, 500 West 120th Street, New York,
NY 10027, USA\\\{{\tt mchudnov,plm2109,js4374,mz2325}\}{\tt
@columbia.edu}\vspace{-5ex}}
\maketitle

\begin{abstract}
We show that the 4-coloring problem can be solved in polynomial time for graphs
with no induced 5-cycle $C_5$ and no induced 6-vertex path $P_6$.
\end{abstract}

\section{Introduction}

Graphs considered in this paper are finite and simple (undirected, no loops,
no parallel edges).  We say that a graph $G$ {\bf contains} a graph $H$ if $G$
has an induced subgraph isomorphic to $H$.

A {\bf coloring} of a graph $G$ assigns labels to the vertices of $G$ in such a
way that no two adjacent vertices receive the same label. To model this more
conveniently, we may assume that the labels come from a fixed set
$\{1,2,\ldots,k\}$ of numbers and we call this a $k$-coloring. 

The {\sc Coloring} problem then asks, given a graph $G$, to find a coloring of
$G$ using minimum number of colors. If the number of colors $k$ is specified, we
call this the {\sc $k$-Coloring} problem which asks to determine if a
$k$-coloring of $G$ exists or not.

{\sc Coloring} is NP-complete in cubic planar graphs
\cite{planar-col}. It admits polynomial time algorithms in many interesting
cases such as in interval graphs, chordal graphs, comparability graphs, and more
generally in perfect graphs \cite{perfect-col}.  Given the diverse nature of
these two extremes, the natural question arises: {\em what distinguishes hard
cases from easy cases?}

One approach to answer this question focuses on studying graphs $G$ that contain
no graph from a fixed collection $\{H_1,H_2,\ldots,H_t\}$ and looking at the
resulting complexity of the problem on such graphs.  Such graphs $G$ are
commonly referred to as $(H_1,H_2,\ldots,H_t)$-free. We write $H$-free if only
one graph $H$ is excluded. For $H$-free graphs, a complete classification is
known \cite{kktw}: {\sc Coloring} is polynomial-time solvable for $H$-free
graphs if $H$ is an induced subgraph of $P_4$ or $P_3+P_1$ (disjoint union of
$P_1$ and $P_3$) and is NP-complete otherwise. Some results excluding two
induced subgraphs are also known \cite{survey-col}.

For {\sc $k$-Coloring} where $k\geq 3$ is fixed, the situation is different.
For $H$-free graphs, the problem is NP-complete for every $k\geq 3$ if $H$
contains a cycle or a claw $K_{1,3}$, and so the remaining cases of interest are
when $H$ is the disjoint unions of paths. It is therefore natural to study the
case when $H$ is the $t$-vertex path $P_t$.  This problem has enjoyed a long and
exciting history of incremental results \cite{p6k3-2, 3col-p6, maria-p7-1,
maria-p7-2, pavol-shenwei, p5-col, shenwei, narrowing, 3col-p6-1, 4col-pk, p6k3,
woeginger-sgall}, since the seminal paper \cite{woeginger-sgall}, and
culminating in the most general case proved recently by Shenwei Huang
\cite{shenwei} who established that {\sc $k$-Coloring} is NP-complete if $k\geq
4$ and $t\geq 7$, and if $k\geq 5$ and $t\geq 6$.  (See \cite{survey-col} for a
survey of these results.)

Unlike the hardness proofs, less is known about polynomial-time cases.  The case
$t\leq 4$ is not interesting, since an optimal coloring can be found in linear
time \cite{cographs}. For $t=5$, a polynomial-time algorithm \cite{p5-col}
exists for every fixed $k$ , but the exponent of the polynomial heavily depends
on $k$.  For $t\geq 6$, because of the above hardness results, there are only
two interesting remaining cases: $t=6$ and $k=4$, and $t\geq 6$ and $k=3$.  The
latter has been established for $t=6$ \cite{3col-p6-1,3col-p6} and very recently
also for $t=7$ \cite{maria-p7-1,maria-p7-2}, while the situation when $t\geq 8$
remains wide open. 

In this paper, we consider the other remaining case $t=6$ and $k=4$. Namely,
{\sc 4-Coloring} in $P_6$-free graphs. Unlike the other outstanding cases, this
problem has actually been conjectured to be polynomial-time solvable
\cite{shenwei}. Here, we use a more fine-grained approach where, in addition, we
exclude a fixed induced $\ell$-cycle $C_\ell$, and consider the problem in
$(P_t,C_\ell)$-free graphs.  This line of research has shown a lot of promise
recently (as a way of bringing us closer to understanding the general case) with
several notable results (see \cite{survey-col}). The cases $k=3$ and
$t\in\{6,7\}$ were both established by first considering $\ell=3$
\cite{p6k3,maria-p7-1}.  For $\ell=4$, {\sc $k$-Coloring} is linear-time
solvable for every $k$ and $t$ \cite{survey-col}, and in some cases, even a
forbidden subgraph characterization is~known~\cite{pavol-shenwei}.  On the
hardness side, {\sc 4-Coloring} is NP-complete when $\ell\geq 5$, $\ell\neq 7$
and $t\geq 7$, or if $\ell=7$~and~$t\geq 9$~\cite{pavol-shenwei}.

\begin{table}\centering
\begin{tabular}{r||cccccccc}
 & $t\leq 4$ & $t=5$ & $t=6$ & $t=7$ & $t=8$ & $t\geq 9$\\
\hline\hline\rule{0pt}{2.5ex}
$k=3$ & P\cite{cographs}  & P\cite{woeginger-sgall}  & P\cite{3col-p6-1} & P\cite{maria-p7-2}  & ? & ?\\
$k=4$ & P\cite{cographs}  & P\cite{p5-col} & {\bf ?} & NPc\cite{shenwei}  & NPc\cite{4col-p8} & NPc\cite{4col-pk}\cite{woeginger-sgall}$^{^\star}$\\
$k=5$ & P\cite{cographs}  & P\cite{p5-col}  & NPc\cite{shenwei} & NPc  & NPc\cite{woeginger-sgall} & NPc\\
$k=6$ & P\cite{cographs}  & P\cite{p5-col}  & NPc & NPc{\cite{3col-p6}}  & NPc & NPc\\
$k\geq 7$ & P\cite{cographs} & P\cite{p5-col} & NPc  & NPc & NPc  & NPc\\[0.5ex]
\hline \rule{0pt}{2.5ex}
$k=\min$ & P\cite{cographs} & NPc\cite{kktw} & NPc  & NPc & NPc  & NPc
\end{tabular}
\quad
\parbox{2cm}{\vskip 3cm $^{^\star}$ for $t=12$}
\caption{Summary of the complexity status of {\sc $k$-Coloring} of $P_t$-free
graphs.\label{fig:kcol}}
\end{table}

\begin{table}\centering
\begin{tabular}{r||cccccc}
$k=4$ & $t\leq 5$ & $t=6$ & $t=7$ & $t=8 $ & $t\geq 9$\\
\hline\hline\rule{0pt}{2.5ex}
$\ell=3$ & P\cite{p5-col} & P\cite{p6k3-2} & P\cite{maria-p7-1} & ?  & ?$^{^\star}$ \\
$\ell=4$ & P\cite{p5-col} & P\cite{survey-col} & P\cite{survey-col} & P\cite{survey-col} & P\cite{survey-col}\\
$\ell=5$ & P\cite{p5-col} & P$^{^{\star\star}}$ & NPc\cite{pavol-shenwei} & NPc & NPc\\
$\ell=6$ & P\cite{p5-col} & ? & NPc\cite{pavol-shenwei} & NPc & NPc\\
$\ell=7$ & P\cite{p5-col} & ? & ? & ? & NPc\cite{pavol-shenwei}\\
$\ell\geq 8$ & P\cite{p5-col} & ? &
NPc\cite{pavol-shenwei} & NPc & NPc
\end{tabular}\qquad
\parbox{4cm}{\vskip 2.5cm $^{^\star}$ NPc for $t=22$
\cite{narrowing}\\$^{^{\star\star}}$
this paper}
\caption{Complexity status of {\sc 4-Coloring} of $(P_t,C_\ell)$-free
graphs.\label{fig:4col}}
\end{table}

Therefore a natural case to consider is {\sc 4-Coloring} when $t=6$ and
$\ell=5$. This is notably so due to the fact that for 4-colorable $P_6$-free
graphs, the 5-cycle is one of only two possible obstructions for perfection --
the other being the 7-antihole.  \medskip

This problem is the main focus of this paper. In particular, we establish that
the problem admits a polynomial-time algorithm. We summarize this as the
following main theorem.

\begin{theorem}\label{thm:main}
There exists a polynomial-time algorithm for {\sc 4-Coloring} 
$(P_6,C_5)$-free graphs.
\end{theorem}

We establish this by describing an algorithm for the problem. Our algorithm is
based on a special cutset decomposition and local consistency testing.  The
resulting algorithm is surprisingly simple, but its correctness rests heavily on
intricate structural properties of the considered graphs.

It is noteworthy to mention that the corresponding list version of the problem
(where vertices, in addition, have lists of allowed colors) has been recently
established to be NP-complete \cite{narrowing}. This is rather surprising, since
often the list version is no more difficult. This tends to be the case whenever
bounded width techniques are used (since these allow for easy extension to
lists). Our algorithm is not based on these techniques.

Our work suggests further possible extensions by weakening the restriction
brought about by the exclusion of the induced 5-cycle, by excluding a larger
induced subgraph such as the 5-wheel instead.  The structural properties we
uncovered seem to provide room for future improvements.

\section{Definitions and Notation}

In this paper, a graph is always finite, simple (undirected, no selfloops, no
parallel edges) and connected. The vertex and edge sets of a graph $G$ are
denoted by $V(G)$ and $E(G)$, respectively. The edge set $E(G)$ consists of
unordered pairs $\{u,v\}$.  For brevity, we write $uv$ to denote the edge
$\{u,v\}$. 

We write $N(v)$ to denote the neighborhood of $v$, i.e., the set of vertices
$u\neq v$ where $uv\in E(G)$. Note that $v\not\in N(v)$. For $X\subseteq V(G)$,
we write $N(X)$ to denote the neighbors of $X$, i.e., the vertices in
$V(G)\setminus X$ with a neighbor in $X$. In other words, $N(X)=\bigcup_{v\in X}
\big(N(v)\setminus X\big)$. 

For $X\subseteq V(G)$, we write $G[X]$ to denote the subgraph of $G$ induced by
$X$, i.e., the graph whose vertex set is $X$ where two vertices are adjacent if
and only if they are adjacent in $G$.  We write $G-X$ to denote the subgraph of
$G$ induced by $V(G)\setminus X$. We write $G-x$ in place of $G-\{x\}$.  A
connected component of $G$ is a maximal set $X$ such that $G[X]$ is connected.

A set $X\subseteq V(G)$ is a {\em clique} if all vertices in $X$ are pairwise
adjacent. A set $X\subseteq V(G)$ is a {\em stable set} or an {\em independent
set} if no two vertices in $X$ are adjacent.  A set $X\subseteq V(G)$ is {\em
complete} to $Y\subseteq V(G)$ if every $x\in X$ is adjacent to every $y\in Y$
(in particular, $X\cap Y=\emptyset$). A set $X\subseteq V(G)$ is {\em
anticomplete} to $Y\subseteq V(G)$ if $X\cap Y=\emptyset$, and there are no
edges in $G$ with one endpoint in $X$ and the other in $Y$.  If $X$ is neither
complete nor anticomplete to $Y$, we say that $X$ is {\em mixed} on $Y$. We say
that $x$ is complete to, anticomplete to, or mixed on $Y$ if $\{x\}$ is complete
to, anticomplete to, or mixed on $Y$, respectively.  

The complement of $G$ is the graph with vertex set $V(G)$ where two vertices are
adjacent if and only if they are not adjacent in $G$.

A $k$-{\em coloring} of $G$ is a mapping $c:V(G)\rightarrow \{1,2,\ldots,k\}$
such that $c(u)\neq c(v)$ for all $uv\in E(G)$.  The chromatic number of $G$ is
the smallest $k$ for which $G$ has a $k$-coloring.

In a {\em partial} coloring we allow some vertices to not have a value assigned
yet. We say that these vertices are {\em uncolored}.  Vertices $v$ of $G$
sometimes come with lists $L(v)\subseteq \{1,\ldots,k\}$ of permitted colors. We
say that a coloring $c$ {\em respects} lists $L$ if $c(v)\in L(v)$ for all $v\in
V(G)$.  

We write $x_1\tp x_2\tp\ldots \tp x_t$ to denote a path in $G$ going through
vertices $x_1,x_2,\ldots,x_t$ in this order.  The path may not be induced.
Similarly, we write $x_1\tp x_2\tp\ldots\tp x_t\tp x_1$ to denote a (not
necessarily induced) cycle in $G$ going through $x_1,x_2,\ldots,x_t$ and back to
$x_1$.  We write $P_k$ to denote the $k$-vertex path, and write $C_k$ to denote
the $k$-vertex cycle. We write $K_k$ to denote the complete graph on $k$
vertices. If a cycle is induced in $G$, it is called a {\em hole}.  An induced
cycle in the complement of $G$ is called an {\em antihole}. An induced
$k$-vertex-cycle in the complement of $G$ is called a $k$-antihole.

A graph $G$ is perfect if for every induced subgraph $H$ of $G$, the size of a
largest clique in $H$ equals the chromatic number of $H$.  Observe that a
$(2k+1)$-antihole requires $k+1$ colors to color.  Odd holes and odd antiholes
are the obstructions to perfection \cite{spgt}.

\section{Overview of the algorithm}

\noindent The algorithm has of 2 phases: {\bf cleaning} and {\bf coloring}.  
The input is a $(P_6,C_5)$-free graph~$G$.

In the {\em cleaning phase}, we run a polynomial-time algorithm ({\bf Phase I
algorithm}) on $G$, that either finds that $G$ is not 4-colorable, or outputs a
{\em clean} $(P_6,C_5)$-free graph $G'$ (notion of clean defined later) such
that $G$ is 4-colorable if and only if $G'$ is, and such that a 4-coloring of
$G'$ can be in polynomial time extended to a 4-coloring of~$G$. 

The algorithm, in particular, finds a so-called {\em chromatic cutset} and
contracts it into a clique.  After this step, $G'$ has no $K_5$ and no
9-antihole.

We then apply to $G'$ the clique-cutset decomposition algorithm of Tarjan
\cite{tarjan}.  This produces (in polynomial time) a linear-size list  $\cal G$
of sugraphs of $G'$ such that no graph in the list has a clique cutset, and $G'$
is 4-colorable if and only if every graph in the list is. 

In the {\em coloring phase}, we test 4-colorability for each graph $G''$ in the
list $\cal G$.  If $G''$ is perfect, then we declare that $G''$ is 4-colorable
(since $G'$ and hence $G''$ has no $K_5$).  To find out if $G''$ is perfect we
only need to check for a 7-antihole.  

Using this antihole and the fact that $G''$ contains no clique cutset, we apply
another polynomial time algorithm ({\bf Phase II algorithm}) which either
outputs a 4-coloring of $G''$ or determines that none exists.  A subroutine of
this algorithm is the list 3-coloring algorithm from \cite{3col-p6}.  If a
4-coloring exists for each $G''$, then these colorings are combined (by
permuting colors and expanding chromatic cutsets) to produce a 4-coloring of
$G$.\smallskip

We summarize the key steps as follows (further details left for later).
\newpage

\begin{itemize}[]
\item {\bf Input:} a $(P_6,C_5)$-free graph $G$

\item {\bf Phase I. cleaning:} iteratively apply Phase I reduction algorithm as long as
needed\\[1ex]
\hspace*{1em}{\em/* now $G$ is clean or it is determined that it is not 4-colorable */}

\item {\bf Clique cutsets:} apply clique cutset decomposition to remove all
clique cutsets\\[1ex]
\hspace*{1em}{\em/* now $G$ is clean and has no clique cutset */}

\item {\bf Perfectness:} test if $G$ is perfect. If so, then $G$ is 4-colorable.\\[1ex]
\hspace*{1em}{\em/* now $G$ is clean, has no clique cutset, and contains a 7-antihole */}

\item {\bf Phase II. coloring:} apply Phase II algorithm to determine
if $G$ is 4-colorable
\end{itemize}

The rest of the paper is structured as follows. First in Section
\ref{sec:antihole}, we derive some important properties of 7-antiholes and their
neighbors. Then in Section \ref{sec:cleaning}, we describe the notion of a clean
graph, desribe the Phase I cleaning procedure, and prove its correctness.
Finally, in Section \ref{sec:coloring}, we describe the Phase II coloring
algorithm and prove its correctness. Both proofs are accompanied by a complexity
analysis to ensure polynomial running time.

\section{Structural properties of antiholes}\label{sec:antihole}

Throughout this section, $G$ shall refer to a $(P_6,C_5)$-free graph containing
a 7-antihole. Let $C=\{v_0,v_1,\ldots,v_6\}$ denote the vertex set of this
antihole where $v_iv_j\in E(G)$ if and only if $2\leq|i-j|\leq 5$.  All index
arithmetic involving the vertices in $C$ shall be be treated as reduced modulo
7.  Vertices $v_i$ and $v_{i+1}$ are said to be {\em consecutive} in $C$.
Following the natural cyclical order of the antihole, for $i<j$, we say that
vertices $v_i,v_{i+1},\ldots,v_{j}$ are {\em consecutive}, and for $i>j$,
vertices $v_i,v_{i+1},\ldots,v_6,v_0,v_1,\ldots,v_{j}$ are {\em
consecutive}.\smallskip

Let $X$ denote the set of all vertices of $G$ with a neighbor in $C$, and let
$Y$ denote the set $V(G)\setminus (C\cup X)$.  The vertices $x$ in $X$ are of
the following three types.\smallskip

\begin{compactitem}
\item If $x$ is adjacent only to $v_i$ among the vertices in $C$, we say that
$x$ is a {\bf leaf} at $v_i$. 
\item If $x$ is complete to a {\em triangle} of $C$ (three pairwise adjacent
vertices in $C$), we say that $x$ is a {\bf big} vertex. 
\item If $x$ is neither a leaf nor a big vertex, we say that $x$ is a {\bf
small} vertex.
\end{compactitem}
\smallskip

\noindent We let $Z$ denote the set of all {\bf leaves}, and let $Z_\ell$ denote
the set of all {\bf leaves at} $v_\ell$.

In what follows, we prove a few useful facts about the structure of the vertices
in $X\cup Y$.

\begin{lemma}The following properties hold for the vertices in $X\cup Y$.\smallskip
\begin{compactenum}[\rm~~(1)]
\item
Let $x\in X$ be a vertex complete to $\{v_i,v_{i+1}\}$ for some $i$. Then
$x$ is also complete to at least one of $v_{i-1},v_{i+2}$.\smallskip
\item
Let $x\in X$ be a small vertex or a leaf. Then one of the following holds:
\begin{compactenum}[\rm({\theenumi}.1)]
\item $x$ is anticomplete to $\{v_{i-1},v_{i+2}\}$ and adjacent to exactly one of 
$v_i,v_{i+1}$ for some~$i$.
\item $x$ has exactly 3 or exactly 4 consecutive neighbors in $C$.
\end{compactenum}\smallskip
\item
Let $\ell\in\{0,1,\ldots,6\}$ and let $x\in X \setminus Z_\ell$ be a vertex with
a neighbor $u\in Y\cup Z_\ell$. Then 
\begin{compactenum}[\rm({\theenumi}.1)]
\item
$x$ is not a leaf.
\item
If $x$ is adjacent to $v_i$ and anticomplete to $\{v_{i-1},v_{i+1},v_{i+2}\}$,
then $u\in Z_\ell$ and $\ell\in\{i,i+1,i+2\}$. 
\item
If $x$ is mixed on an edge $uv$ where $v\in Y\cup Z_\ell$, then $x$
has at least 5 consecutive neighbors in $C$.
\end{compactenum}\smallskip
\item
Let $x\in X$ be a small vertex with a neighbor in $Y\cup Z$. Then
\begin{compactenum}[\rm({\theenumi}.1)]
\item $x$ has exactly 3 or exactly 4 consecutive neighbors in $C$, and
\item if $v_j\not\in N(x)$, then
\begin{compactenum}[\hspace*{-2em}\rm({\theenumi.\theenumii}.1)]
\item if $x$ has a neighbor in $Z_j$, then $x$ is complete to
$\{v_{j+2},v_{j+3},v_{j-3},v_{j-2}\}$,
\item if $Z_j\neq\emptyset$, then $x$ is anticomplete to
$\{v_{j-1},v_j,v_{j+1}\}$.
\end{compactenum}
\end{compactenum}\smallskip
\item
Let $x_1,x_2 \in X$ be non-adjacent small vertices, both with a neighbor in
$Y\cup Z$.\\  Then $N(x_1)\cap C \subseteq N(x_2)\cap C$ or $N(x_2)\cap
C\subseteq N(x_1)\cap C$.
\end{compactenum}

\end{lemma}

\begin{proof}
In the following text, we prove each of \ref{clm:zero}-\ref{clm:y2}
individually.

\begin{claim}\label{clm:zero}
Let $x\in X$ be a vertex complete to $\{v_i,v_{i+1}\}$ for some $i$. Then $x$ is
also complete to at least one of $v_{i-1},v_{i+2}$.
\end{claim}

\begin{proofclaim}
If $x$ is anticomplete to $\{v_{i-1},v_{i+2}\}$, then $x\tp v_i\tp v_{i+2}\tp
v_{i-1}\tp v_{i+1}\tp x$ is an induced 5-cycle in $G$.
\end{proofclaim}

\begin{claim}\label{clm:attach}
Let $x\in X$ be a small vertex or a leaf. Then one of the following holds:
\begin{compactenum}[\rm({\theclaim}.1)]
\item $x$ is anticomplete to $\{v_{i-1},v_{i+2}\}$ and adjacent to exactly one of 
$v_i,v_{i+1}$ for some~$i$.\setreflabel{(\theclaim.\theenumi)}\label{clm:3.1}
\item $x$ has exactly 3 or exactly 4 consecutive neighbors in $C$.
\setreflabel{(\theclaim.\theenumi)}\label{clm:3.2}
\end{compactenum}
\end{claim}

\begin{proofclaim}
Since $x$ is in $X$, it has at least one neighbor in $C$. Since $x$ is not a big
vertex, it is not complete to $C$.  Thus we may assume by symmetry that
$v_0\not\in N(x)$ but $v_1\in N(x)$.   If $x$ is non-adjacent to $v_3$, then by
\ref{clm:zero}, $x$ is also non-adjacent to $v_2$, and we obtain outcome
\ref{clm:3.1} for $i=1$.  Thus we may assume that $x$ is adjacent to $v_3$.
This means that $x$ is anticomplete to $\{v_5,v_6\}$, since otherwise $x$ is
complete to a triangle of $C$ implying that $x$ is big, but we assume otherwise.
If now $x$ is adjacent to $v_2$, then we obtain outcome \ref{clm:3.2}.  Thus we
may assume that $x$ is non-adjacent to $v_2$, and therefore we obtain outcome
\ref{clm:3.1} for $i=0$.

\end{proofclaim}

\begin{claim}
Let $\ell\in\{0,1,\ldots,6\}$ and let $x\in X \setminus Z_\ell$ be a vertex with
a neighbor $u\in Y\cup Z_\ell$. Then 
\begin{compactenum}[\rm({\theclaim}.1)]
\item
$x$ is not a leaf.
\setreflabel{(\theclaim.\theenumi)}\label{clm:y1-1}
\item
If $x$ is adjacent to $v_i$ and anticomplete to $\{v_{i-1},v_{i+1},v_{i+2}\}$,
then $u\in Z_\ell$ and $\ell\in\{i,i+1,i+2\}$. 
\setreflabel{(\theclaim.\theenumi)}\label{clm:y1-1.5}
\item
If $x$ is mixed on an edge $uv$ where $v\in Y\cup Z_\ell$, then $x$
has at least 5 consecutive neighbors in $C$.
\setreflabel{(\theclaim.\theenumi)}\label{clm:y1-2}
\end{compactenum}
\end{claim}

\begin{proofclaim}
First we prove \ref{clm:y1-1.5}. Assume that $x$ is adjacent to $v_i$ and
anticomplete to $v_{i-1},v_{i+1},v_{i+2}$. Then since $u\tp x\tp v_i\tp
v_{i+2}\tp v_{i-1}\tp v_{i+1}$ is not an induced $P_6$ in $G$, it follows that
$u\in Z_\ell$ and $\ell\in\{i-1,i,i+1,i+2\}$. If $\ell=i-1$, then $u\tp x\tp
v_i\tp v_{i+2}\tp v_{i-1}\tp u$ is an induced 5-cycle in $G$, a contradiction.
This proves \ref{clm:y1-1.5}.
\smallskip

For \ref{clm:y1-1}, suppose that $x\in Z_i$ for some $i$.  Then $x$ is adjacent
to $v_i$ and anticomplete to $\{v_{i-2},$ $v_{i-1},$ $v_{i+1},v_{i+2}\}$. Using
\ref{clm:y1-1.5} twice, we conclude that $\ell=i$. But $x\in X\setminus Z_\ell$,
a contradiction. This proves \ref{clm:y1-1}.\smallskip

For \ref{clm:y1-2}, suppose that $x$ is mixed on an edge $uv$ where $v\in Y\cup
Z_\ell$. By \ref{clm:y1-1}, $x$ is not a leaf, since it is adjacent to $u\in
Y\cup Z_\ell$.  Thus $x$ has at least two non-consecutive neighbors in $C$ by
\ref{clm:zero}. Observe further by \ref{clm:y1-1} that $u,v$ are either both in
$Z_\ell$ or both in $Y$, since $uv\in E(G)$.  For contradiction, assume that $x$
does not have 5 consecutive neighbors in $C$.

An edge $v_iv_j$ where $\{v_i,v_j\}\subseteq N(x)$ is {\em good} if
$\{v_i,v_{i-1}\}$ is complete to $\{v_j,v_{j+1}\}$, and  $x$ is anticomplete to
$\{v_{i-1},v_{j+1}\}$.  Suppose there exists a good edge $v_iv_j$. Then by
symmetry (between $i$ and~$j$), we may assume that $\ell\not\in\{j,j+1\}$.  If
$v_{i-1}$ is anticomplete to $\{u,v\}$, then $v\tp u\tp x\tp v_j\tp v_{i-1}\tp
v_{j+1}$ is an induced $P_6$ in $G$, a contradiction.  Thus $\ell=i-1$ and
$u,v\in Z_\ell$.  But now $v_{i-1}\tp u\tp x\tp v_i\tp v_{j+1}\tp v_{i-1}$ is an
induced 5-cycle in $G$, a contradiction.  This shows that no good edge exists.

Suppose that $x$ is complete to $\{v_i,v_{i+1}\}$ for some $i$.  By
\ref{clm:zero}, $x$ is adjacent to at least one $v_{i-1},v_{i+2}$. By symmetry,
we may assume that $v_{i+2}\in N(x)$. Since $x$ does not have 5 consecutive
neighbors, we may further assume by symmetry that $v_{i-1}\not\in N(x)$ and for
some $j\in\{i+2,i+3\}$, we have $v_j\in N(x)$ and $v_{j+1}\not\in N(x)$. But now
$v_iv_j$ is a good edge, impossible.

This shows that $x$ is not complete to any two consecutive vertices in $C$.
Since is mixed on $C$, it follows that $x$ is complete to $v_i$ and anticomplete
to $\{v_{i-1},v_{i+1},v_{i+2}\}$ for some $i$. Since $x$ has at least two
non-consecutive neighbors in $C$,  there exists $j\in\{i+3,i+4,i+5\}$ such that
$v_j\in N(x)$ and $v_{j-1}\not\in N(x)$.  But now $v_jv_i$ is a good edge, a
contradiction. This proves \ref{clm:y1-2}.

\end{proofclaim}

\begin{claim}\label{clm:y1-3}
Let $x\in X$ be a small vertex with a neighbor in $Y\cup Z$. Then
\begin{compactenum}[\rm({\theclaim}.1)]
\item $x$ has exactly 3 or exactly 4 consecutive neighbors in $C$, and
\setreflabel{(\theclaim.\theenumi)}\label{clm:6.1}
\item if $v_j\not\in N(x)$, then
\setreflabel{(\theclaim.\theenumi)}\label{clm:6.2}
\begin{compactenum}[\hspace*{-2em}\rm({\theclaim.\theenumi}.1)]
\item if $x$ has a neighbor in $Z_j$, then $x$ is complete to
$\{v_{j+2},v_{j+3},v_{j-3},v_{j-2}\}$,
\setreflabel{(\theclaim.\theenumi.\theenumii)}\label{clm:6.2.1}
\item if $Z_j\neq\emptyset$, then $x$ is anticomplete to
$\{v_{j-1},v_j,v_{j+1}\}$.
\setreflabel{(\theclaim.\theenumi.\theenumii)}\label{clm:6.2.2}
\end{compactenum}
\end{compactenum}
\end{claim}

\begin{proofclaim}
For \ref{clm:6.1}, let $y\in Y\cup Z$ be a neighbor of $x$. Thus $y\in Y\cup
Z_\ell$ for some $\ell$.  Since $x$ is a small vertex, we may assume by
\ref{clm:attach} that for contradiction \ref{clm:3.1} holds.  Namely, we may
assume by symmetry that $x$ is adjacent to $v_1$ and anticomplete to
$\{v_0,v_2,v_3\}$.  From \ref{clm:y1-1.5}, we deduce that $y\in Z_\ell$ and
$\ell\in\{1,2,3\}$.  Since $x$ is not a leaf, there exists $j\in\{4,5,6\}$ such
that $v_j\in N(x)$ and $v_{j+1}\not\in N(x)$. Consider largest such $j$.  If
$j=4$, then the maximality of $j$ implies that $x$ is anticomplete to
$\{v_5,v_6\}$ and so we deduce by \ref{clm:y1-1.5} that $\ell\in\{4,5,6\}$, a
contradiction.  If $j=5$, then $x$ is non-adjacent to $v_4$ by \ref{clm:zero},
and so we deduce $\ell\in\{5,6,0\}$ by \ref{clm:y1-1.5}, a contradiction.  It
follows that $v_6\in N(x)$.  If $v_4\in N(x)$, then $x$ is complete to a
triangle $v_1,v_4,v_6$ of $C$ which is impossible, since $x$ is a small vertex.
Thus $v_4\not\in N(x)$ and so $v_5\not\in N(x)$ by \ref{clm:zero}. But now
$\ell\in\{4,5,6\}$ by \ref{clm:y1-1.5}, a contradiction.  This proves
\ref{clm:6.1}.  \smallskip

For \ref{clm:6.2.1}, assume that $v_j\not\in N(x)$ and that $x$ has a neighbor
$y\in Z_j$.  Suppose first that $v_{j+1}\in N(x)$. Then $v_{j-2}\not\in N(x)$,
since by \ref{clm:6.1}, $x$ has exactly 3 or 4 consecutive neighbors in $C$.
Thus $y\tp v_j\tp v_{j-2}\tp v_{j+1}\tp x\tp y$ is an induced 5-cycle in $G$, a
contradiction. We therefore conclude that $v_{j+1}\not\in N(x)$ and by symmetry,
also $v_{j-1}\not\in N(x)$.  If also $v_{j+2}\not\in N(x)$, then $x\tp y\tp
v_j\tp v_{j+2}\tp v_{j-1}\tp v_{j+1}$ is an induced $P_6$ in $G$, a
contradiction. Thus $v_{j+2}\in N(x)$ and by symmetry, also $v_{j-2}\in N(x)$.
Now since the neighbors of $x$ in $C$ are consecutive, it follows that $x$ is
complete to $\{v_{j+2}, v_{j+3}, v_{j-2}, v_{j-3}\}$ as claimed.  This proves
\ref{clm:6.2.1}.  \smallskip

Finally for \ref{clm:6.2.2}, assume that $v_j\not\in N(x)$ and there exists
$z\in Z_j$. Recall that $x$ has a neighbor $y\in Y\cup Z$.  Thus by
\ref{clm:6.1} and \ref{clm:6.2.1}, we may assume that $x$ is not adjacent to
$z$, and $y\in Y\cup Z_\ell$ where $\ell\neq j$.  Now for contradiction, assume
by symmetry that $v_{j+1}\in N(x)$.  Suppose first that $yz\in E(G)$.  If $y\in
Y$, then $x$ is mixed on the edge $yz$ of $Y\cup Z_j$, but that contradicts
\ref{clm:y1-2}, since $x$ is a small vertex. If $y\in Z_\ell$, then $y$ a leaf
in $X\setminus Z_j$ and has a neighbor in $Z_j$, contradicting \ref{clm:y1-1}.
Thus $y$ is not adjacent to $z$, and moreover, $x$ is not adjacent to $v_{j-2}$,
since $x$ has exactly 3 or 4 consecutive neighbors in $C$ by \ref{clm:6.1}.  Now
if $y$ is anticomplete to $\{v_{j-2},v_j,v_{j+1}\}$, then $z\tp v_j\tp
v_{j-2}\tp v_{j+1}\tp x\tp y$ is an induced $P_6$ in $G$, a contradiction. Thus
$y\in Z_\ell$ and $\ell\in\{j-2,j+1\}$. If $\ell=j+1$, then $z\tp v_j\tp
v_{j+2}\tp v_{j-1}\tp v_{j+1}\tp y$ is an induced $P_6$ in $G$. Therefore we
deduce that $\ell=j-2$.  In other words, $x$ has a neighbor in $Z_{j-2}$, namely
$y$, but $v_{j-2}\not\in N(x)$, which by \ref{clm:6.2.1} implies that $x$ is
complete to $\{v_j,v_{j+1},v_{j+2},v_{j+3}\}$. However, $v_j\not\in N(x)$, a
contradiction. This proves \ref{clm:6.2.2}.

\end{proofclaim}

\begin{claim}\label{clm:y2}
Let $x_1,x_2 \in X$ be non-adjacent small vertices, both with a neighbor in
$Y\cup Z$.\\  Then $N(x_1)\cap C \subseteq N(x_2)\cap C$ or $N(x_2)\cap
C\subseteq N(x_1)\cap C$.
\end{claim}

\begin{proofclaim}
Let $y_1$ be a neighbor of $x_1$ in $Y\cup Z$, and let $y_2$ be a neighbor of
$x_2$ in $Y\cup Z$. Therefore there exist $k$ and $\ell$ such that $y_1\in Y\cup
Z_k$ and $y_2\in Y\cup Z_\ell$ (possibly $k=\ell$ or $y_1=y_2$).  If possible
choose $y_1,y_2$ so that $y_1=y_2$.

\begin{innerclaim}

\begin{claim}[\theclaim.1]\label{clm:7.1}
If $y_1\neq y_2$, then $\{x_1,y_1\}$ is anticomplete to $\{x_2,y_2\}$.
\end{claim}

\begin{proofclaim}
If $y_1\neq y_2$, then our choice of $y_1,y_2$ implies that the two vertices
cannot be chosen to be equal (otherwise we would have chosen them equal).
Therefore $x_1y_2,x_2y_1\not\in E(G)$. It remains to show that $y_1y_2\not\in
E(G)$.  Assume, for contradiction, that $y_1y_2\in E(G)$.  Then $y_1,y_2\in
Y\cup Z_\ell$ by \ref{clm:y1-1}.  But now $x_1$ is mixed on the edge $y_1y_2$ of
$Y\cup Z_\ell$, which contradicts \ref{clm:y1-2}, since $x_1$ is a small vertex.
\end{proofclaim}

\end{innerclaim}

We say that an edge $v_iv_j\in E(G)$ is {\em good} if $v_i\in N(x_1)\setminus
N(x_2)$ and $v_j\in N(x_2)\setminus N(x_1)$. To prove \ref{clm:y2}, suppose for
contradiction that the claim of \ref{clm:y2} is false.

\begin{innerclaim}

\begin{claim}[\theclaim.2]\label{clm:7.2}
There exists a good edge.
\end{claim}

\begin{proofclaim}
Since the claim is false, there exist $v_i\in N(x_1)\setminus N(x_2)$ and
$v_j\in N(x_2)\setminus N(x_1)$. If $v_iv_j\in E(G)$, then $v_iv_j$ is a good
edge as claimed. Thus we may assume that $v_iv_j\not\in E(G)$.  Since
$v_iv_j\not\in E(G)$, we may assume by symmetry that $j=i+1$.  By \ref{clm:6.1},
both $x_1$ and $x_2$ have exactly 3 or 4 consecutive neighbors in $C$. Thus
$v_{i+2}\not\in N(x_1)$, since $v_{i+1}\not\in N(x_1)$ but $v_i\in N(x_1)$.
Also $v_{i+2}\in N(x_2)$, since $v_{i+1}\in N(x_2)$ but $v_i\not\in N(x_2)$.
But now $v_iv_{i+2}$ is a good edge as claimed. 
\end{proofclaim}

\end{innerclaim}

By \ref{clm:7.2}, we may therefore assume that there exists a good edge
$v_iv_j\in E(G)$, where $v_i\in N(x_1)\setminus N(x_2)$ and $v_j\in
N(x_2)\setminus N(x_1)$.  Suppose first that $\{y_1,y_2\}$ is anticomplete to
$\{v_i,v_j\}$.  If $y_1=y_2$, then $y_1\tp x_1\tp v_i\tp v_j\tp x_2\tp y_1$ is
an induced 5-cycle in $G$, a contradiction. Thus $y_1\neq y_2$ which implies, by
\ref{clm:7.1}, that $\{x_1,y_1\}$ is anticomplete to $\{x_2,y_2\}$.  But now
$y_1\tp x_1\tp v_i\tp v_j\tp x_2\tp y_2$ is an induced $P_6$ in $G$, a
contradiction.  This shows that $y_1\in Z_k$ and $k\in\{i,j\}$, or that $y_2\in
Z_\ell$ and $\ell\in\{i,j\}$. By symmetry, let us assume the former, namely that
$y_1\in Z_k$ and $k\in\{i,j\}$. This leads to two cases $k=i$ or $k=j$.

Suppose first that $k=i$. Recall that $v_i\not\in N(x_2)$ and
$Z_i\neq\emptyset$, since $y_1\in Z_k=Z_i$. This implies by \ref{clm:6.2.2} that
$x_2$ is anticomplete to $\{v_{i-1},v_i,v_{i+1}\}$. Since by \ref{clm:6.1} $x_2$
has exactly 3 or 4 consecutive neighbors in $C$, we deduce that $x_2$ is
complete to $\{v_{i-3},v_{i+3}\}$ and one or both of $v_{i-2},v_{i+2}$.  Now
recall that $v_i\in N(x_1)$. Since also $x_1$ has exactly 3 or 4 consecutive
neighbors, we may assume by symmetry that $v_{i+1}\in N(x_1)$ and
$v_{i-3}\not\in N(x_1)$.  Thus if $y_1=y_2$, then $y_1\tp x_1\tp v_{i+1}\tp
v_{i-3}\tp x_2\tp y_1$ is an induced 5-cycle in $G$. Therefore $y_1\neq y_2$ and
so $\{x_1,y_1\}$ is anticomplete to $\{x_2,y_2\}$ by \ref{clm:7.1}.  Recall that
$y_2\in Y\cup Z_\ell$.  Thus if $y_2\in Y$ or if $y_2\in Z_\ell$ but
$\ell\not\in\{i+1,i-3\}$, then $y_1\tp x_1\tp v_{i+1}\tp v_{i-3}\tp x_2\tp y_2$
is an induced $P_6$ in $G$. We therefore conclude that $y_2\in Z_\ell$ and
$\ell\in\{i+1,i-3\}$.  If $\ell=i+1$, then $y_1\tp v_i\tp v_{i+2}\tp v_{i-1}\tp
v_{i+1}\tp y_2$ is an induced $P_6$ in $G$, a contradiction.  Therefore
$\ell=i-3$ and we note that $Z_{i-3}\neq\emptyset$, since $y_2\in
Z_\ell=Z_{i-3}$. Recall that $v_{i-3}\not\in N(x_1)$. So by \ref{clm:6.2.2} we
have that $x_1$ is anticomplete to $\{v_{i-2},v_{i-3},v_{i+3}\}$.  But now we
deduce that $y_1\tp x_1\tp v_{i+1}\tp v_{i+3}\tp x_2\tp y_2$ is an induced $P_6$
in $G$, a contradiction.

We may therefore assume that $k=j$. Recall that $v_j\not\in N(x_1)$ and $y_1$ is
a neighbor of $x_1$ in $Z_k=Z_j$. Thus by \ref{clm:6.2.1} and \ref{clm:6.1} we
conclude that $x_1$ is complete to $\{v_{j+2},v_{j+3},v_{j-3},v_{j-2}\}$ and
anticomplete to $\{v_{j-1},v_j,v_{j+1}\}$.  Recall that $v_j\in N(x_2)$. By
\ref{clm:6.1}, $x_2$ has exactly 3 or 4 consecutive neighbors in $C$. Thus by
symmetry we may assume that $v_{j+1}\in N(x_2)$ and $v_{j-3}\not\in N(x_2)$.  If
$y_1=y_2$, then $y_1\tp x_1\tp v_{j-3}\tp v_{j+1}\tp x_2\tp y_1$ is an induced
5-cycle in $G$.  Thus $y_1\neq y_2$ and by \ref{clm:7.1}, $\{x_1,y_1\}$ is
anticomplete to $\{x_2,y_2\}$.  So if $y_2\in Y$ or if $y_2\in Z_\ell$ but
$\ell\not\in\{j+1,j-3\}$, then $y_2\tp x_2\tp v_{j+1}\tp v_{j-3}\tp x_1\tp y_1$
is an induced $P_6$ in $G$. We therefore conclude that $y_2\in Z_\ell$ and
$\ell\in\{j+1,j-3\}$. If $\ell=j+1$, then $y_1\tp v_j\tp v_{j+2}\tp v_{j-1}\tp
v_{j+1}\tp y_2$ is an induced $P_6$ in $G$. Therefore $\ell=j-3$ and we see that
$Z_{j-3}\neq\emptyset$ because $y_2\in Z_\ell=Z_{j-3}$. Recall that
$v_{j-3}\not\in N(x_2)$. Thus by \ref{clm:6.2.2}, we deduce that $x_2$ is
anticomplete to $\{v_{j+3},v_{j-3},v_{j-2}\}$.  But now $y_2\tp x_2\tp
v_{j+1}\tp v_{j+3}\tp x_1\tp y_1$ is an induced $P_6$ in $G$, a contradiction.

\end{proofclaim}

This completes the proof of the lemma.
\end{proof}

\section{Phase I. - Cleaning}\label{sec:cleaning}

Let $G$ be a $(P_6,C_5)$-free graph. We say that a 7-antihole $C$ of $G$ is {\em clean} if
\smallskip

\begin{compactenum}[(C1)]
\item no vertex is complete to $C$,
\setreflabel{(C\theenumi)}\label{enum:c1}
\item if two vertices each have 6 neighbors in $C$, then they have the same neigbors in $C$,
\setreflabel{(C\theenumi)}\label{enum:c2}
\item for every connected component $K$ of $G-\big(C\cup
N(C)\big)$ at least one of the following~holds:
\begin{compactenum}[(C\theenumi.1)]
\item $N(K)$ is a clique,
\setreflabel{(C\theenumi.\theenumii)}\label{enum:c3.1}
\item there exists a vertex complete to $K$.
\setreflabel{(C\theenumi.\theenumii)}\label{enum:c3.2}
\end{compactenum}
\setreflabel{(C\theenumi)}\label{enum:c3}
\end{compactenum}
\smallskip

We say that $G$ is {\em clean} if $G$ contains no $K_5$ and no 9-antihole, and
every 7-antihole of $G$ is clean.

\begin{lemma}\label{lem:cleaning}
There exists a polynomial-time algorithm that given a connected $(P_6,C_5)$-free
graph $G$ either
\begin{compactitem}
\item finds that $G$ is not 4-colorable, or
\item finds that $G$ is clean, or
\item outputs a smaller connected $(P_6,C_5)$-free graph $G'$ such that $G$ is
4-colorable iff $G'$ is, and such that a 4-coloring of $G'$ can be in polynomial time
extended to a 4-coloring~of~$G$.
\end{compactitem}
\end{lemma}

We shall refer to the algorithm from the above lemma as {\bf Phase I algorithm}.
For the proof we need to discuss a special type of cutset that we shall be
dealing with.

\subsection{Chromatic cutset}

Let $G$ be a graph.  A set $S\subseteq V(G)$ is said to be {\em 1-chromatic} if
$S$ is an independent set of $G$ such that in every 4-coloring of $G$ the
vertices of $S$ receive the same color. In particular, every vertex of $G$
constitutes a 1-chromatic set.

It seems natural to identify the vertices of a 1-chromatic set $S$ into a single
vertex. Clearly, this does not change the chromatic number of the graph.
However, it may create long induced paths or cycles. In order to make sure this
does not happen, we need to be more careful when we apply such a transformation.
When $S$ is part of a special cutset, this can be achieved as described
below.\smallskip

A set $S\subseteq V(G)$ is a {\em cutset} of $G$ if $G$ is connected but $G-S$
is disconnected. A {\em clique cutset} is a cutset that is also a clique.  A
cutset of $G$ is a {\em chromatic cutset} if it admits a partition into
1-chromatic sets that are pairwise complete to each other.  Observe that a
clique cutset is a special case of a chromatic cutset.  A cutset $S$ is said to
be a {\em minimal separator} if there are distinct connected components $K,K'$
of $G-S$ such that every vertex of $S$ has both a neighbor in $K$ and a neighbor
in $K'$. Note that a minimal separator is not necessarily an inclusion-wise
minimal cutset. We need this distinction for convenience.

By a slight abuse of terminology, we shall say that a set $S\subseteq V(G)$ is
{\em contracted} to a vertex $s$ to mean the operation of replacing $S$ in $G$
by a new vertex $s$ whose set of neighbors is precisely $N(S)$.

\begin{lemma}\label{lem:contract}
Let $G$ be a connected $(P_6,C_5)$-free graph. Suppose that $S$ contains a
chromatic cutset $S$ that is also a minimal separator.  Let
$(S_1,S_2,\ldots,S_t)$ be a partition of $S$ into pairwise complete 1-chromatic
sets. Let $G'$ be obtained from $G$ by contracting each set $S_i$ into a new
vertex $s_i$.  Then $G'$ is also a $(P_6,C_5)$-free graph.
\end{lemma}

\begin{proof}
Since $S$ is a minimal separator, there exist  distinct connected components $K$
and $K'$ of $G-S$ such that each vertex of $S$ has both a neighbor in $K$ and a
neighbor in $K'$.

For contradiction, suppose that $G'$ contains $H$ where $H$ is $P_6$ or a
5-cycle.  If $H$ is disjoint from $\{s_1,s_2,\ldots,s_t\}$, then $H$ is an
induced path or cycle of $G$, a contradiction.  Thus $H$ contains at least one
of $s_1,s_2,\ldots,s_t$.  In fact, $H$ contains exactly one or two of these
vertices, since they are pairwise adjacent.

Suppose first by symmetry that $H$ contains $s_1,s_2$.  For $i\in\{1,2\}$, let
$A_i$ denote the set of neighbors of $s_i$ in $V(H)\setminus\{s_1,s_2\}$. Note
that $|A_i|\leq 1$, since $s_1$ is adjacent to $s_2$. This implies that there
exists $y_i\in S_i$ complete to $A_i$. In particular, $y_1$ is adjacent to
$y_2$.  Consequently, $V(H) \setminus \{s_1,s_2\} \cup \{y_1,y_2\}$ induces a
$P_6$ or a 5-cycle in $G$, a contradiction.

So we may assume that $V(H)\cap \{s_1,s_2,\ldots,s_t\}=\{s_1\}$.  If there
exists $y_1\in S_1$ complete to all neighbors of $s_1$ in $H$, then
$V(H)\setminus\{s_1\}\cup\{y_1\}$ induces a $P_6$ or a 5-cycle in $G$.  Thus we
may assume that this is not the case.  This implies that $s_1$ has exactly two
neighbors $a,b$ where $a,b\in V(G)$ and there are distinct vertices $y_a,y_b\in
S_1$ such that $ay_a,by_b\in E(G)$ and $ay_b,by_a\not\in E(G)$.  Note that
$y_ay_b\not\in E(G)$, since $S_1$ is a stable set.  Thus if $H$ is a 5-cycle,
then $V(H)\setminus\{s_1\}\cup\{y_a,y_b\}$ is an induced $P_6$ of $G$, a
contradiction.

Therefore $H$ is a path $x_1\tp x_2\tp x_3\tp x_4\tp x_5\tp x_6$ where by
symmetry we may assume $s_1=x_k$ for $k\in\{2,3\}$, and where $a=x_{k-1}$ and
$b=x_{k+1}$.  Observe that the path $x_{k+1}\tp x_{k+2}\tp\ldots\tp x_6$
contains none of the vertices $s_1,\ldots,s_t$. It therefore lies completely in
some connected component $K''$ of $G-S$.  So $K''\neq K$ or $K''\neq K'$. By
symmetry, we may assume that $K''\neq K$.  Recall that $y_a,y_b\in S$ and that
each vertex of $S$ has a neighbor in $K$.  Thus there exists a path  from $y_a$
to $y_b$ in $G[K\cup \{y_a,y_b\}]$.  Let $P'$ be a shortest such path.  Since
$S_1$ is a stable set, the path $P'$ has at least 3 vertices.  Also, $P'$ is an
induced path, since it is a shortest path.  Observe that the internal vertices
of $P'$ are in $K$ and so they are anticomplete to $\{x_{k+1},\ldots,x_6\}$.
Thus, since $k\in\{2,3\}$, this shows that $P'\tp x_{k+1}\tp x_{k+2}\tp\ldots\tp
x_6$ is an induced path of length $\geq 6$, which therefore contains an induced
$P_6$.

This completes the proof.
\end{proof}\vspace{-2ex}

\subsection{Proof of Lemma \ref{lem:cleaning}}

Let $n=|V(G)|$. First we test if $G$ contains a $K_5$ or a 9-antihole. If so, we
stop and report that $G$ is not 4-colorable. Note that this is correct, since
neither $K_5$ nor 9-antihole are 4-colorable.  Thus we may assume that $G$
contains no $K_5$ and no 9-antihole.

Next we test if $G$ is clean. Namely, we check if every 7-antihole of $G$ is
clean. Each time we process an antihole $C$, we first check the conditions
\ref{enum:c1} and \ref{enum:c2} in time $O(n^2)$.  If one of these conditions
fails, we stop and report that $G$ is not 4-colorable.  Indeed, if \ref{enum:c1}
fails, then there is a vertex complete to $C$ which forces $C$ to be colored
using at most 3 colors in any 4-coloring of $G$, but this is impossible.  If
\ref{enum:c2} fails, then there are vertices $x_1,x_2$ each with 6 neighbors in
$C$, where the neighborhoods of $x_1,x_2$ in $C$ are different. Namely,
$N(x_1)\cap C=C\setminus\{v_i\}$ and $N(x_2)\cap C=C\setminus\{v_j\}$ for some
$i\neq j$. If there is a 4-coloring of $G$, then necessarily $x_1$ and $v_i$
have the same color $c_1$, and no other vertex of $C$ is colored with color
$c_1$.  Similarly, $x_2$ and $v_j$ have the same color $c_2$ where no other
vertex of $C$ is colored with $c_2$. Since $i\neq j$, we deduce that $c_1\neq
c_2$ and so $C\setminus\{v_i,v_j\}$ is colored using 2 colors, but this is
impossible, since it contains a triangle.

If the conditions \ref{enum:c1}, \ref{enum:c2} hold, we continue by computing in
time $O(n^2)$ the connected components of $G-(C\cup N(C))$ . This produces
$O(n)$ components $K$ to check. For each such $K$, we test conditions
\ref{enum:c3.1} and \ref{enum:c3.2} in time $O(n^2)$.  Altogether, this takes
polynomial time.  

Having done this for all 7-antiholes, if we determine that $G$ is clean, we stop
and report this.  If not, we find a 7-antihole $C=\{v_0,v_1,\ldots,v_6\}$ and a
connected component $K$ of $G-(C\cup N(C))$ such that $N(K)$ is not a clique and
no vertex in $V(G)\setminus K$ is complete to $K$.  Let $S=N(K)$.  Note that
$S\subseteq N(C)$.

In what follows, we shall use this to produce a smaller graph $G'$ with the
required properties.  Since we have not rejected $G$ earlier, we shall assume in
the subsequent text that the conditions \ref{enum:c1} and \ref{enum:c2} hold.

\begin{claim}\label{clm:x1}
Every $u\in S$ has at least 5 consecutive neighbors in $C$.
\end{claim}

\begin{proofclaim}
Since $u$ has a neighbor in $K$ but no vertex is complete to $K$, we deduce that
$u$ is mixed on an edge of $K$.  Since $S\subseteq N(C)$, it follows by
\ref{clm:y1-2} that $x$ has at least 5 consecutive neighbors in $C$.
\end{proofclaim}

\begin{claim}\label{clm:x2}
Let $u_1,u_2\in S$. Then $u_1u_2\not\in E(G)$ if and only if $N(u_1)\cap
C\subseteq N(u_2)\cap C$ or $N(u_2)\cap C\subseteq N(u_1)\cap C$.
\end{claim}

\begin{proofclaim}
By \ref{clm:x1}, the neighbors of each of $u_1,u_2$ in $C$ are consecutive and
each has at least 5 neighbors in $C$.  Thus, by symmetry, we may assume that
$|N(u_1)\cap C|\geq |N(u_2)\cap C|\geq 5$.

This implies that if $N(u_2)\cap C\subseteq N(u_1)\cap C$, then $\{u_1,u_2\}$ is
complete to a triangle of $C$. From this we conclude that $u_1u_2\not\in E(G)$,
since $G$ contains no $K_5$. This proves the backward direction of the claim.

For the forward direction, assume that $u_1u_2\not\in E(G)$.  By \ref{enum:c1},
we deduce that $|N(u_2)\cap C|\neq 7$. If $|N(u_2)\cap C|=6$, then $N(u_2)\cap
C=N(u_1)\cap C$ by \ref{enum:c2}. Thus we may assume $|N(u_2)\cap C|=5$.

Observe that there exists $v_\ell\in N(u_1)\cap N(u_2)$.  For contradiction,
assume that there exists $v_i\in N(u_1)\setminus N(u_2)$ and $v_j\in
N(u_2)\setminus N(u_1)$. We choose $i,j$ so that $v_iv_j\in E(G)$.  Suppose that
this is not possible. Then by symmetry $j=i+1$. This implies by \ref{clm:x1}
that $u_1$ is complete to $C\setminus\{v_{i+1},v_{i+2}\}$.  Similarly, $u_2$ is
anticomplete to $\{v_i,v_{i-1}\}$, since $|N(u_2)\cap C|=5$.  But then we could
choose $v_{i-1}$ in place of $v_i$, a contradiction.

This shows that we may assume that $v_iv_j\in E(G)$.  Now recall that both $u_1$
and $u_2$ have neighbors in $K$. Let $P=x_1\tp x_2\tp\ldots\tp x_k$ be a
shortest path from $u_1$ to $u_2$ whose internal vertices are in $K$.  Note that
$k\geq 3$, since $u_1,u_2$ are non-adjacent.  Moreover, $k\leq 5$, since
otherwise $x_1\tp\ldots\tp x_6$ is an induced $P_6$ in $G$. If $k=5$, then
$v_i\tp x_1\tp\ldots\tp x_5$ is an induced $P_6$ in $G$.  If $k=4$, then $v_k\tp
x_1\tp x_2\tp x_3\tp x_4\tp v_k$ is an induced 5-cycle in $G$.  Thus $k=3$ in
which case $v_i\tp x_1\tp x_2\tp x_3\tp v_j\tp v_i$ is an induced 5-cycle in
$G$, a contradiction.

\end{proofclaim}

Now we discuss two possibilities.

\begin{claim}\label{clm:x3}
If every vertex in $S$ has exactly 5 neighbors on $C$, then $S$ is a chromatic
cutset.
\end{claim}

\begin{proofclaim}
Assume that every vertex in $S$ has exactly 5 neighbors on $C$. We partition $S$
into equivalence classes $S_1,S_2,\ldots S_t$ where two vertices $u_1,u_2\in S$
are equivalent if and only if $C\setminus N(u_1)=C\setminus N(u_2)$. It follows
that each $S_i$ is 1-chromatic, since $N(S_i)\cap C$ contains a triangle, and
$G$ has no $K_5$.  Moreover, by \ref{clm:x2}, each $S_i$ is complete to every
$S_j$ where $i\neq j$. This shows that $S$ is a chromatic cutset as claimed.
\end{proofclaim}

\begin{claim}\label{clm:x4}
If there exist $u\in S$ with at least 6 neighbors on $C$, then $S$ is a
chromatic cutset
\end{claim}

\begin{proofclaim}
Let $u$ be a vertex in $S$ with at least 6 neighbors on $C$.  By \ref{enum:c1},
we may assume that $C\setminus N(u)=\{v_0\}$.  

This implies that in every 4-coloring of $G$ the vertices of the antihole $C$
are uniquely colored (up to permuting colors).  Namely, $v_0$ uses color 1,
$v_1,v_2$ use color 2, $v_3,v_4$ color 3, and $v_5,v_6$ color 4.  Define $S_i$
to be the set of vertices $u\in S$ such that no vertex in $N(u)\cap C$ is
colored with color $i$.  It follows that every vertex of $S$ belongs to at most
one $S_i$, since by \ref{clm:x1} it has at least 5 neighbors in $C$.

Suppose that some vertex $v$ of $S$ belongs to no $S_i$. This implies that $v$
is adjacent to $v_0$ and at least one of $v_1,v_2$, one of $v_3,v_4$, and one of
$v_5,v_6$. Therefore  $\{u,v\}$ is complete to a triangle of $C$, and so
$uv\not\in E(G)$, since $G$ contains no $K_5$. From this, we deduce by
\ref{clm:x2} that $N(v)\cap C\supseteq N(u)\cap C$ because $v_0\in N(v)\setminus
N(u)$. Thus $N(v)\supseteq C$, since $C\setminus N(u)=\{v_0\}$. But now we
violate \ref{enum:c1}.  This proves that the sets $S_i$ partition $S$.

We now show that each $S_i$ is a stable set.  If $v,w\in S_2$, then
$C\setminus\{v_1,v_2\}$ are precisely all neighbors of $v$ and $w$ in $C$, since
both $v$ and $w$ have at least 5 consecutive neighbors in $C$ by \ref{clm:x1}.
We therefore conclude that $vw\not\in E(G)$ by \ref{clm:x2}.  Similarly if
$v,w\in S_3$ or $v,w\in S_4$. If $v,w\in S_1$ and $vw\in E(G)$, then it follows
by \ref{clm:x2} and by symmetry that $N(v)\cap C=\{v_1,\ldots,v_5\}$ and
$N(w)\cap C=\{v_2,\ldots,v_6\}$.  Thus $uw,uv\not\in E(G)$ by \ref{clm:x2},
since $N(u)\cap C=\{v_1,\ldots,v_6\}$.  But now
$\{v_1,v_2,v_3,v_4,v_5,v_6,v,u,w\}$ induces a 9-antihole in $G$, a
contradiction. This proves that no such $v,w$ exist and thus each $S_i$ is
stable.

Finally, recall again that each $v\in S_i$ has at least 5 consecutive neighbors
in $C$, but is not adjacent to any vertex of $C$ of color $i$, by the
construction of $S_i$. This implies that for every color $j\neq i$, the vertex
$v$ is adjacent to at least one vertex of $C$ of color $j$.  This implies that
each $S_i$ is 1-chromatic, and that vertices from distinct sets $S_i$ have
incomparable neighborhoods in $C$. Therefore by \ref{clm:x2} such vertices are
adjacent, which proves that the sets $S_1,S_2,S_3,S_4$ are pairwise complete.
This proves that $S$ is a chromatic cutset as claimed.

\end{proofclaim}

It follows from \ref{clm:x3} and \ref{clm:x4} that $S=N(K)$ is a chromatic
cutset. Observe that $S$ is a minimal separator, since every $v\in S$ has both a
neighbor in $K$ and a neighbor in $C$.  In particular, $S$ induces a complete
multipartite graph with parts $S_1,\ldots,S_t$. We can determine such a
partition in time $O(n^2)$.  Then in time $O(n^2)$, we contract each $S_i$ into
a vertex. By Lemma \ref{lem:contract}, the resulting graph $G'$ does not contain
a $P_6$ or a 5-cycle. Recall that the component $K$ fails \ref{enum:c3.1},  and
so $S$ is not a clique. Therefore $G'$ has fewer vertices than $G$.  Finally,
since the process of constructing $G'$ involves only contracting 1-chromatic
sets, it follows that $G$ is 4-colorable if and only if $G'$ is, and a
4-coloring of $G'$ can be directly extended to $G$.  Therefore we stop and
output $G'$ as the result of the algorithm. This is correct, as the graph $G'$
has the required properties of the lemma.  The total complexity is clearly
polynomial.

This completes the proof of Lemma \ref{lem:cleaning}.\hfill\qed

\section{Phase II. - Coloring}\label{sec:coloring}

In this section, we show how to test if a given clean $(P_6,C_5)$-graph with an
induced 7-antihole and no clique cutset is 4-colorable. We refer to this
procedure as {\bf Phase II algorithm}.

\subsection{Algorithm}

Let $G$ be a clean $(P_6,C_5)$-free graph with no clique cutset.  Let
$C=\{v_0,v_1,\ldots,v_6\}$ denote the vertex set of a 7-antihole of $G$. Let
$X=N(C)$ and let $Y=V(G)\setminus (C\cup X)$. Following the terminology of
Section \ref{sec:antihole}, the vertices of $X$ are big, small, and leaves. Let
$S$ denote the set of all small vertices having a neighbor in $Y\cup Z$.\medskip

\begin{compactenum}[\bf Step 1.]
\item
Initialize the set $R$ to be empty, and then for each
$i\in\{0,1,\ldots,6\}$ do the following.\smallskip
\begin{compactitem}[\textbullet]
\item If $S$ contains a small vertex complete to $\{v_i,v_{i+1},v_{i+2}\}$ and
anticomplete to $v_{i-1}$, then put one such a vertex $u$ into $R$.  If possible
choose $u$ so that $v_{i+3}\not\in N(u)$.
\end{compactitem}

\smallskip

\item Consider a 4-coloring $c$ of $G[C\cup R]$. (If none exists, output 
``$G$ is not 4-colorable''.)\\[0.5ex] Initialize $L(v)=\{c(v)\}$ for
each $v\in C\cup R$, and $L(v)=\{1,2,3,4\}$ for every other~$v$.\smallskip

\item
\uline{Propagate the colors:} apply the following rule as long as it 
changes the lists~$L$.
\begin{compactitem}[\textbullet]
\item If there exists $v$ with $L(v)=\{i\}$, remove $i$ from $L(u)$ for each
$u\in N(v)$.
\end{compactitem}
\smallskip

\item \uline{Test neighborhoods:} for every vertex $v$, initialize the set
$Q(v)=\emptyset$ and then for each color $i\in L(v)$ do the following.\smallskip
\begin{compactitem}[\textbullet]
\item Set $L'(w)\leftarrow L(w)\setminus\{i\}$ for each $w\in N(v)$.
\item Test if there exists a 3-coloring of $G[N(v)]$ respecting the lists $L'$.
\item If the coloring does not exist, add $i$ to $Q(v)$.
\end{compactitem}
\smallskip

\item \uline{Test common neighborhoods:} for each pair of distinct vertices
$(u,v)$, initialize the set $Q(u,v)=\emptyset$ and then for each color $i\in
L(u)$ and $j\in L(v)$ such that $i\neq j$,\smallskip
\begin{compactitem}[\textbullet]
\item Set $L'(w)\leftarrow L(w)\setminus\{i,j\}$ for each $w\in N(u)\cap N(v)$.
\item Test if there exists a 2-coloring of $G[N(u)\cap N(v)]$ respecting the lists $L'$.
\item If the coloring does not exist, add $(i,j)$ to $Q(u,v)$.
\end{compactitem}
\smallskip

\item Test if there exists a 4-coloring $c$ of $G'=G[C\cup (X\setminus Z)]$
satisfying the following.
\begin{compactitem}[\textbullet]
\item For all $u\in V(G')$, if $c(u)=i$, then $i\in L(u)$ and $i\not\in Q(u)$.
\item For all distinct $u,v\in V(G')$, if $c(u)=i$ and $c(v)=j$, then
$(i,j)\not\in Q(u,v)$.
\end{compactitem}
\smallskip

If the coloring exists, then output ``$G$ is 4-colorable''.
\smallskip

\item If all possible colorings chosen in Step 2 fail the test of Step 6, 
then output ``$G$ is not 4-colorable''.
\end{compactenum}
\medskip

Note that we chose to distinguish lists $L$ and $Q$ for convenience of
subsequent proof arguments.

\subsection{Nice Coloring}

For the proof of correctness of the algorithm, we need to discuss a special type
of (partial) coloring of~$G$.  We say that a partial coloring $c$ of $G$ is a
{\em nice coloring} if $c$ is a partial 4-coloring of $G$ using colors
$\{1,2,3,4\}$~where:\smallskip

\begin{compactenum}[(N1)]
\item the set of uncolored vertices is stable,
\setreflabel{(N\theenumi)}\label{enum:n1}
\item the neighborhood of every uncolored vertex $y$ admits a partition into 
sets $U$,~$W$~where\setreflabel{(N\theenumi)}\label{enum:n2}
\begin{compactitem}[\textbullet]
\item $c(u)\in\{1,2\}$ for all $u\in U$,
\item $c(w)\in\{3,4\}$ for all $w\in W$, and
\item at least one of $U$, $W$ is an independent set (or empty).
\end{compactitem}
\end{compactenum}
\smallskip

Note that the sets $U$, $W$ are uniquely fixed by $c$, and can possibly be
empty.

\begin{lemma}\label{lem:nice}
Let $G$ be a $(P_6,C_5)$-free graph.
If $G$ admits a nice coloring, then $G$ is 4-colorable.
\end{lemma}

\begin{proof}
For contradiction, suppose that the claim is false.  Namely, suppose that there
exists a nice coloring $c$ of $G$, but $G$ is not 4-colorable. Choose $c$ to
have smallest possible number of uncolored vertices.  Clearly, $c$ has at least
one uncolored vertex, otherwise $c$ is a 4-coloring of $G$ which we assume does
not exist.

Consider an uncolored vertex $y$. Since $c$ is a nice coloring, all neighbors of
$y$ are colored, and there is a partition of $N(y)$ into two sets $U$, $W$, one
of which is an independent set or empty, such that $c(u)\in\{1,2\}$ for each
$u\in U$ and $c(w)\in\{3,4\}$ for each $w\in W$.  By symmetry, assume that $U$
is an independent set~or~empty.

Let $U_1$ denote the set of vertices $u\in U$ with $c(u)=1$ and let $U_2$ denote
the set of vertices $u\in U$ with $c(u)=2$.  Recall that $c(w)\in\{3,4\}$ for
all $w\in W$.  Thus if $U_1$ or $U_2$ is empty, then we color $y$ by 1 or 2,
respectively, and the resulting coloring is again a partial 4-coloring of $G$.
Moreover, it is a nice coloring, since the properties \ref{enum:n1} and
\ref{enum:n2} are not affected by giving $y$ a color.  However, this coloring
has fewer uncolored vertices than $c$, a contradiction.

We may therefore assume that both $U_1$ and $U_2$ are non-empty.  Let $V_{12}$
denote the set of all vertices $v\in V(G)$ with $c(v)\in\{1,2\}$.  Recall that
$U$ is an independent set, and note that $U\subseteq V_{12}$.  Suppose first
that there exists path in $G[V_{12}]$ from a vertex of $U_1$ to a vertex of
$U_2$.  Let $P=x_1\tp x_2\tp\ldots\tp x_t$ be a shortest such path, where
$x_1\in U_1$ and $x_t\in U_2$.  This implies that $c(x_1)=1$ and $c(x_t)=2$.
Note that $x_2\not\in U$, since $x_1\in U$ and $U$ is an independent set of $G$.
Also, $x_2\not\in W$, since $x_2\in V_{12}$ implying $c(x_2)\in\{1,2\}$, while
every $w\in W$ has $c(w)\in\{3,4\}$.  So $x_2\not\in N(y)$. Similarly,
$x_{t-1}\not\in N(y)$.  This shows that $t\geq 3$.  Moreover, since $c$ is a
coloring, the colors 1 and 2 alternate on the path $P$, which shows that $t$ is
even.  Now note that $P$ is an induced path of $G$, since it is a shortest path
in $G[V_{12}]$.  Thus if $t=4$, then $y\tp x_1\tp x_2\tp x_3\tp x_4\tp y$ is an
induced 5-cycle in $G$, while if $t\geq 6$, then $P$ is an induced path of
length at least 6, and so $G$ contains a $P_6$.

This shows that there is no such a path.  Consequently, let $D$ denote the union
of all connected component $K$ of $G[V_{12}]$ such that $K\cap
U_1\neq\emptyset$. Note that $U_1\subseteq D$ and since $G[V_{12}]$ contains no
path from $U_1$ to $U_2$, we have $U_2\cap D=\emptyset$.  We therefore modify
$c$ by switching the colors 1 and 2 on the vertices in $D$.  Note that the
resulting partial coloring of $G$ is still a partial 4-coloring, and also a nice
coloring, since this operation does not affect the properties \ref{enum:n1} and
\ref{enum:n2}. Of course, the colors in the neighbouhoods of other uncolored
vertices might have changed, but only vertices with colors 1 or 2 changed their
colors to 2 or 1, respectively, and so the coloring remains nice.  Now all
vertices in $U$ are using color 2 and so we color $y$ by 1 to obtain another
nice coloring.  But this coloring has fewer uncolored vertices than $c$, a
contradiction.

This completes the proof.
\end{proof}

\subsection{Correctness}

We are now ready to prove that Phase II algorithm correctly decides if $G$ is
4-colorable.

\begin{lemma}\label{lem:correct}
Let $G$ be a clean $(P_6,C_5)$-free graph with a 7-antihole and no clique
cutset. Then $G$ has a 4-coloring if and only if Phase II algorithm reports
that $G$ is 4-colorable.
\end{lemma}

\begin{proof}
Suppose that $G$ has a 4-coloring $\zeta$. Consider the iteration of Steps 2-6
where the coloring of $G[C\cup R]$ was fixed according to $\zeta$.  Since
$\zeta$ is a coloring of $G$, for every vertex $v$, the color $\zeta(v)$ is
never removed from the list $L(v)$ in Step 3. Similarly, $\zeta$ restricted to
$N(v)$ yields a 3-coloring of $G[N(v)]$ not using the color $\zeta(v)$. Thus
$\zeta(v)$ is never added to the list $Q(v)$ in Step 4.  Similarly, the pair
$(\zeta(u),\zeta(v))$ is never added to $Q(u,v)$ in Step 5. Thus the test in
Step 6 succeds, namely $\zeta$ restricted to $C\cup (X\setminus Z)$ provides the
required coloring.  This shows that Phase II algorithm correctly answers that
$G$ is 4-colorable.

Now, for the converse, we say that $c$ is a {\em good coloring}, if it satisfies
the conditions of Step 6.  Assume that Phase II algorithm answers that $G$ is
4-colorable. This happens during some iteration of the loop in Steps 2-6 after
the algorithm finds a good coloring of $G'=G[C\cup (X\setminus Z)]$ in Step 6.
Let $c$ be this coloring.  Note that this coloring coincides on $C\cup R$ with
the coloring chosen in Step 2.  This follows from the construction of the lists
$L(v)$ in Step 2, and since $c(v)\in L(v)$ for all $v\in V(G')$ because $c$ is a
good coloring. We show how to extend $c$ to a nice coloring of $G$. The claim
will then follow by Lemma \ref{lem:nice}.

For the rest of the proof, recall the notation: $X=N(C)$, $Y=V(G)\setminus
(C\cup X)$, $Z$ denotes the set of all leaves (with respect to $C$), $S$ the set
of small vertices having a neighbor in $Y\cup Z$, and $Z_\ell$ the set of
leaves at $v_\ell$. Also, let $B$ denote the set of all big vertices.  Note that
$X\supseteq B\cup S\cup Z$.

For brevity, we shall write $R$, $L(v)$, $Q(v)$, and $Q(u,v)$ to mean the final
values of these sets when the algorithm terminates (not their intermediate
values, unless specified otherwise). In particular, note that the lists $L(v)$
are initialized at Step 2, modified at Step 3, and never modified after Step 3.
Similarly, the sets $R$, $Q(v)$ and $Q(u,v)$ are constructed at Steps 1, 4 and
5, respectively, and never changed at any other step.

We need to establish several useful properties as follows.

\begin{claim}\label{clm:z1}
$L(v)\neq\emptyset$ for all $C\cup (X\setminus Z)$, and if
$|L(v)|=1$, then $L(v)=\{c(v)\}$.
\end{claim}

\begin{proofclaim}
Recall that after Step 3, the lists $L(v)$ do not change any more.  Since $c$ is
a good coloring of $G[C\cup(X\setminus Z)]$, we have $c(v)\in L(v)$, and thus
$L(v)\neq\emptyset$. So if $|L(v)|=1$, it follows that $L(v)=\{c(v)\}$.
\end{proofclaim}

\begin{claim}\label{clm:z2}
$L(v)=\{c(v)\}$ for each $v\in B\cup C\cup R$, and $|L(x)|\leq 2$
for each $x\in X\setminus Z$.
\end{claim}

\begin{proofclaim}
First consider $v\in C\cup R$. After initializing $L(v)$ in Step 2, we have
$L(v)=\{c(v)\}$.  Thus since $L(v)\neq\emptyset$ after Step 3 by \ref{clm:z1},
it follows that $L(v)=\{c(v)\}$ also after Step 3, as claimed.  Next, consider
$v\in B$. Note that $v$ is big and thus there exists a triangle $v_i,v_j,v_k\in
N(v)$. Since $v_i,v_j,v_k\in C$, we have $L(v_i)=\{c(v_i)\}$,
$L(v_j)=\{c(v_j)\}$, and $L(v_k)=\{c(v_k)\}$, as we just proved. Therefore, the
colors $c(v_i)$, $c(v_j)$, $c(v_k)$ were removed from $L(v)$ at Step 3.  Since
$v_i,v_j,v_k$ form a triangle and since $c$ is a coloring, these 3 colors are
pairwise distinct. Therefore $|L(v)|\leq 1$ which by \ref{clm:z1} implies that
$L(v)=\{c(v)\}$ as claimed.

Finally, consider $x\in X\setminus Z$. Since $x$ is not in $Z$, it has at least
two neighbors $v_i,v_j$ in $C$.  If $v_i,v_j$ cannot be chosen to be adjacent,
then it follows by symmetry that $j=i+1$ and $x$ is anticomplete to
$\{v_{i-1},v_{i+2}\}$. But now $x\tp v_i\tp v_{i+2}\tp v_{i-1}\tp v_{i+1}\tp x$
is an induced 5-cycle in $G$, a contradiction.  Thus we may assume that $v_iv_j$
is an edge.  Since $v_i,v_j\in C$, we have $L(v_i)=\{c(v_i)\}$ and
$L(v_j)=\{c(v_j)\}$ as we just proved.  Therefore the colors $c(v_i)$ and
$c(v_j)$ were removed from $L(x)$ at Step 3.  Since $v_i,v_j$ are adjacent and
since $c$ is a coloring, the two colors are distinct. Thus after Step 3, we have
$|L(x)|\leq 2$ as claimed.
\end{proofclaim}

Based on \ref{clm:z2}, let $S^+$ denote the set of all $v\in S$ with $|L(v)|=2$,
and let $S^-=S\setminus S^+$.

\begin{claim}\label{clm:z3}
For all $u,v\in S^+$, if $uv\not\in E(G)$ or if $\{u,v\}$ is complete to adjacent
vertices of $C$, then $L(u)=L(v)$.
\end{claim}

\begin{proofclaim}
First we prove that if $uv\not\in E(G)$, then $\{u,v\}$ is complete to adjacent
vertices of $C$.

Assume $uv\not\in E(G)$. Then by \ref{clm:y2}, we have up to symmetry  that
$N(u)\cap C\subseteq N(v)\cap C$. By \ref{clm:y1-3}, each of $u$, $v$ has
exactly 3 or exactly 4 consecutive neighbors in $C$. Thus $\{u,v\}$ is complete
to at least 3 consecutive vertices of $C$, and so it is complete to some
adjacent vertices of $C$, as claimed.

Now it is enough to show $L(u)=L(v)$ assuming that $\{u,v\}$ is complete to
$\{v_i,v_j\}$ where $v_iv_j\in E(G)$.  Since $u,v\in S^+$, we have
$|L(u)|=|L(v)|=2$.  Since $v_i,v_j\in C$, we have $L(v_i)=\{c(v_i)\}$ and
$L(v_j)=\{c(v_j)\}$ by \ref{clm:z2}.  Thus the colors $c(v_i)$ and $c(v_j)$ were
removed from both $L(u)$ and $L(v)$ at Step 3. Since $v_i$ is adjacent to $v_j$,
the two colors are different, because $c$ is a coloring. Moreover, since
$|L(u)|=|L(v)|=2$, no other colors were removed from the two lists at Step 3.
Thus $L(u)=L(v)=\{1,2,3,4\}\setminus\{c(v_i),c(v_j)\}$ as claimed.

\end{proofclaim}

\begin{claim}\label{clm:z4}
For each $u\in S^+$, there exists $u'\in R$ such that 

\begin{compactenum}[\rm(\theclaim.1)]
\item $c(u')\in L(u)$, and
\setreflabel{(\theclaim.\theenumi)}\label{clm:z4-1}
\item for all $v\in S^+$, either $L(v)=L(u)$ or $v\in N(u')$.
\setreflabel{(\theclaim.\theenumi)}\label{clm:z4-2}
\end{compactenum}
\end{claim}

\begin{proofclaim}
Since $u$ is in $S$, it is a small vertex with a neighbor in $Y\cup Z$. Thus by
\ref{clm:y1-3}, $u$ has exactly 3 or 4 consecutive neighbors in $C$. By
symmetry, we may assume that $u$ is complete to $\{v_1,v_2,v_3\}$ and
anticomplete to $\{v_0,v_5,v_6\}$.  Since $u\not\in R$, there exists $u'\in R$
such that  $u'$ is complete to $\{v_1,v_2,v_3\}$ and anticomplete~to~$v_0$.

Note that $\{u,u'\}$ is complete to $\{v_1,v_3\}$.  By \ref{clm:z2}, we have
$L(v_1)=\{c(v_1)\}$ and $L(v_3)=\{c(v_3)\}$. Thus the colors $c(v_1)$ and
$c(v_3)$ were removed from $L(u)$ at Step 3. The two colors are different, since
$v_1v_3$ is an edge and $c$ is a coloring. By the same argument, $c(u')$ is
neither $c(v_1)$ nor $c(v_3)$.  Thus since $|L(u)|=2$ because $u\in S^+$, it
follows that $L(u)=\{1,2,3,4\}\setminus \{c(v_1),c(v_3)\}$. Therefore $c(u')\in
L(u)$ as claimed.

Now for contradiction, let $v\in S^+$ such that $L(v)\neq L(u)$ and $v\not\in
N(u')$. Since $\{u,u'\}$ is complete to $\{v_1,v_2,v_3\}$, we deduce by
\ref{clm:z3} that $v$ is not complete to $\{v_1,v_3\}$.  Therefore $N(v)\cap
C\subseteq N(u')\cap C$ by \ref{clm:y2}, since $u',v\in S$ and $u'v\not\in
E(G)$.  The fact that $u',v\in S$ also implies by \ref{clm:y1-3} that each of
$u'$, $v$ has exactly 3 or 4 consecutive neighbors in $C$.  Thus since $u'$ is
complete to $\{v_1,v_2,v_3\}$, and anticomplete to $v_0$, and since $v$ is not
complete to $\{v_1,v_3\}$, and $N(v)\cap C\subseteq N(u')\cap C$, it follows
that $N(u')\cap C=\{v_1,v_2,v_3,v_4\}$ and $N(v)\cap C=\{v_2,v_3,v_4\}$.  Now
the choice of $u'$ implies that $v_4\in N(u)$. Indeed, if $v_4\not\in N(u)$,
then $u$ would have been prefered instead of $u'$ to be added to $R$ when
considering vertices of $S$ complete to $\{v_1,v_2,v_3\}$ and anticomplete to
$v_0$.  But now $\{u,v\}$ is complete to $\{v_2,v_4\}$ implying $L(u)=L(v)$ by
\ref{clm:z3}, a contradiction.

\end{proofclaim}

\begin{claim}\label{clm:z5}
For all $u,v,w\in S^+$, if $L(u)$, $L(v)$, $L(w)$ are pairwise distinct, then
\begin{compactenum}[\rm (\theclaim.1)]
\item $L(u)\cup L(v)\cup L(w)=\{1,2,3,4\}$, and
\setreflabel{(\theclaim.\theenumi)}\label{clm:z5-1}
\item $L(u)\cap L(v)\cap L(w)=\emptyset$.
\setreflabel{(\theclaim.\theenumi)}\label{clm:z5-2}
\end{compactenum}
\end{claim}

\begin{proofclaim}
Suppose that the lists $L(u)$, $L(v)$, $L(w)$ are pairwise distinct.

First assume that $L(u)\cup L(v)\cup L(w)\neq\{1,2,3,4\}$.  Since $u,v,w\in
S^+$, we have $|L(u)|=|L(v)|=|L(w)|=2$.  Thus up to permuting colors,
$L(u)=\{1,2\}$, $L(v)=\{1,3\}$, and $L(w)=\{2,3\}$.  By \ref{clm:z4}, there
exists $u'\in R$ such that $c(u')\in L(u)$ and we have $v,w\in N(u')$, because
neither $L(v)$ nor $L(w)$ is equal to $L(u)$.  Since $u'\in R$, we have by
\ref{clm:z2} that $L(u')=\{c(u')\}$.  Thus since $v,w\in N(u')$, the color
$c(u')$ was removed from the lists $L(v)$ and $L(w)$ at Step 3. However, recall
that $c(u')\in L(u)=\{1,2\}$.  Thus if $c(u')=1$, then $1\not\in L(v)$ and if
$c(u')=2$, then $2\not\in L(w)$, a contradiction.  This proves \ref{clm:z5-1}.

For \ref{clm:z5-2}, assume that $L(u)\cap L(v)\cap L(w)\neq\emptyset$.  Again we
have $|L(u)|=|L(v)|=|L(w)|=2$, and so up to permuting colors, we may assume
$L(u)=\{1,2\}$, $L(v)=\{1,3\}$, and $L(w)=\{1,4\}$.  Since $u$ is in $S$, it is
a small vertex with a neighbor in $Y\cup Z$.  Thus by \ref{clm:y1-3}, $u$ has
exactly 3 or 4 consecutive neighbors on $C$.  By symmetry, we may assume that
$u$ is complete to $\{v_1,v_2,v_3\}$ and anticomplete to $\{v_5,v_6\}$ (and
possibly adjacent to $v_0$ or $v_4$). Note that by \ref{clm:z2}, we have
$L(v_i)=\{c(v_i)\}$ for $i\in\{1,2,3\}$.  Thus the colors $c(v_i)$,
$i\in\{1,2,3\}$ were removed from $L(u)$ at Step 3.  In particular, since
$v_1v_3\in E(G)$ and since $c$ is a coloring, we have $c(v_1)\neq c(v_3)$. So
since $L(u)=\{1,2\}$, we may assume by symmetry and up to permuting colors that
$c(v_1)=c(v_2)=3$ and $c(v_3)=4$.  This implies that $v$ is anticomplete to
$\{v_1,v_2\}$, since otherwise the color $c(v_1)=3$ would have been removed from
$L(v)$ at Step 3, but $3\in L(v)$.  By the same argument, $w$ is not adjacent to
$v_3$, since $4\in L(w)$.

Now since $v\in S$, the vertex $v$ has exactly 3 or 4 consecutive neighbors in
$C$. Since $v_1,v_2\not\in N(v)$, we deduce that $v_5\in N(v)$.  Note that $v_5$
is complete to $\{v_1,v_3\}$. Thus $c(v_5)\in \{1,2\}$, since $c(v_1)=3$,
$c(v_2)=4$, and $c$ is a coloring.  Also $L(v_5)=\{c(v_5)\}$. Thus $c(v_5)$ was
removed from $L(v)$ at Step 3, since $v_5\in N(v)$. Since $L(v)=\{1,3\}$, it
therefore follows that $c(v_5)=2$.  From this and the fact that $c$ is a
coloring, we deduce that $c(v_0)=1$, because $v_0$ is complete to
$\{v_2,v_3,v_5\}$.  Thus $w\not\in N(v_0)$, since otherwise $c(v_0)=1$ would
have been removed from $L(w)$, but $1\in L(w)$.  From this it follows that $w$
is complete to $\{v_4,v_5,v_6\}$ and anticomplete to $\{v_0,v_1,v_2,v_3\}$,
because $w$, since $w\in S$, has exactly 3 or 4 consecutive neighbors in $C$,
and is adjacent neither to $v_3$ nor to $v_0$.

Therefore the colors $c(v_4)$ and $c(v_6)$ were removed from $L(w)$ at Step 3.
Since $L(w)=\{1,4\}$, this implies $c(v_4),c(v_6)\in\{2,3\}$.  However, neither
$c(v_4)$ nor $c(v_6)$ can be 3, since they are both adjacent to $v_1$, and
$c(v_1)=3$. Thus $c(v_4)=c(v_6)=2$, but then $c$ is not a coloring, a
contradiction.
\end{proofclaim}

\begin{claim}\label{clm:z6}
For all $u,v,w\in S^+$, the lists $L(u)$, $L(v)$, $L(w)$ are not pairwise
distinct.
\end{claim}

\begin{proofclaim}
Suppose that the lists $L(u)$, $L(v)$, $L(w)$ are pairwise distinct.  Then
$L(u)\cup L(v)\cup L(w)=\{1,2,3,4\}$ by \ref{clm:z5}. Since $u,v,w\in S^+$, we
have $|L(u)|=|L(v)|=|L(w)|=2$.  Thus by symmetry (up to permuting colors) we may
assume $L(u)=\{1,4\}$, $2\in L(v)$, and $3\in L(w)$.  By \ref{clm:z4}, there
exists $u'\in R$ such that $c(u')\in L(u)$ and we have $v,w\in N(u')$, since the
lists $L(u)$, $L(v)$, $L(w)$ are diferent.  Similarly, there exists $v'\in R$
such that $c(v')\in L(v)$ and $u,w\in N(v')$, and we have $w'\in R$ where
$c(w')\in L(w)$ and $u,v\in N(w')$.

Since $u',v',w'\in R$, we have by \ref{clm:z2} that $L(u')=\{c(u')\}$,
$L(v')=\{c(v')\}$, $L(w')=\{c(w')\}$.  Thus the colors $c(v')$ and $c(w')$ were
removed from $L(u)$ at Step 3.  In other words, $c(v'),c(w')\not\in
L(u)=\{1,4\}$.  We similarly deduce that $c(v')\neq 3$ and $c(w')\neq 2$,
because $3\in L(w)$, $2\in L(v)$, and  $w\in N(v')$, $v\in N(w')$.  This shows
that $c(v')=2$ and $c(w')=3$.  Therefore at Step 3, the color 2 was removed from
$L(w)$, and the color 3 was removed from $L(v)$. To summarize, we deduce that
$3\not\in L(v)$ and $2\not\in L(w)$.

Now recall that $c(u')\in L(u)=\{1,4\}$. By symmetry, we may assume $c(u')=1$.
Thus the color 1 was removed at Step 3 from the lists $L(v)$ and $L(w)$, since
$v,w\in N(u')$. This implies only one possibility, namely that $L(v)=\{2,4\}$
and $L(w)=\{3,4\}$, since $2\in L(v)$, $3\in L(w)$, and both lists have size 2
because $v,w\in S^+$. But now $4\in L(u)\cap L(v)\cap L(w)$, contradicting
\ref{clm:z5}. 

\end{proofclaim}

Now back to the proof of the converse direction of Lemma \ref{lem:correct}.  We
construct a nice coloring of $G$ by extending $c$ to some connected components
of $G[Y\cup Z]$.  This is done as follows.  Let $K$ be a connected component of
$G[Y\cup Z]$ that is not colored.  Let $A^+=S^+\cap N(K)$ and $A^-=S^-\cap
N(K)$.  Let $A=A^+\cup A^-$.  Note that $|L(a)|=1$ for all $a\in A^-$ by
\ref{clm:z2} and the definition of $S^-$.

By \ref{clm:y1-2}, $A$ is complete to $K$.  Moreover, by \ref{clm:y1-1}, either
$K$ is a connected component of $G[Y]$, or there exists
$\ell\in\{0,1,\ldots,6\}$ such that $K$ is a connected component of $G[Z_\ell]$.
It therefore follows that the neighborhood of $K$ consists of $A$, of $B\cap
N(K)$, and possibly of $v_{\ell}$ if $K\subseteq Z_\ell$.

Now, if $A^+$ is non-empty, then pick any $x\in A^+$. If $A^+=\emptyset$ and
$K\subseteq Z_\ell$, then let $x=v_\ell$. In $A^+=\emptyset$ and $K\subseteq Y$,
then recall that $C$ is a clean antihole, and $K$ is a connected component of
$G-N(C\cup N(C))$. Thus one of \ref{enum:c3.1}, \ref{enum:c3.2} holds for $K$.
If \ref{enum:c3.1} holds, then $N(K)$ is a clique cutset of $G$, but we assume
that $G$ contains no such a cutset.  Thus \ref{enum:c3.2} holds and so there
exists a vertex in $X\setminus Z$ complete to $K$.

This shows that in all cases, there exists a vertex $x\in C\cup(X\setminus Z)$
complete to $K$, and such that $x\in A^+$ if $A^+\neq\emptyset$.  Up to renaming
colors, we may assume $c(x)=1$.

First, suppose that all vertices in $A^+$ are assigned the color 1 by $c$, i.e.,
$c(v)=1$ for all $v\in A^+$.  Since $c(x)=1$ and since $c$ is a good coloring,
we have $1\not\in Q(x)$. Thus $1$ was not added to $Q(x)$ at Step 4.  In
particular, since the test at Step 4 succeeded, there exists a 3-coloring $c_K$
of $G[K]$ using colors $\{2,3,4\}$, where $c_K(v)\in L(v)$ for each $v\in K$.
This follows from the fact that $K\subseteq N(x)$.  We use this coloring to
extend $c$ to the component $K$. We claim that this does not create a
monochromatic edge. Suppose otherwise, and let $uv$ be such a edge with
$c(u)=c(v)$. Clearly, not both $u,v$ are in $K$, since $c_K$ is a coloring. Also
not both $u,v$ are outside $K$, since $c$ is a (partial) coloring outside $K$.
Thus we may assume that $v\in K$ and $u\in A\cup B\cup C$.  If $u\in A^+$, then
$c(u)=1$, but $c(v)\neq 1$ since the coloring $c_K$ does not use the color $1$.
If $u\in A^-\cup B\cup C$, then by \ref{clm:z1} and \ref{clm:z2}, we have
$L(u)=\{c(u)\}$.  Therefore $c(u)$ was removed from $L(v)$ at Step 3.  But then
$c(u)\neq c(v)$, since $c(v)=c_K(v)\in L(v)$, a contradiction.

So we may now assume that $A^+\neq\emptyset$ and (up to permuting colors) there
exists $y\in A^+$ with $c(y)=2$.  Note that now $x\in A^+$ and $x\neq y$.
Suppose that $c(v)\in\{1,2\}$ for every other $v\in A^+$.  Since $c$ is a good
coloring, we conclude that $(1,2)\not\in Q(x,y)$.  In particular, the pair
$(1,2)$ was not added to $Q(x,y)$ at Step 5, which implies that there exists a
2-coloring $c_K$ of $G[K]$ using colors $\{3,4\}$ such that $c_K(v)\in L(v)$ for
each $v\in K$.  For this, recall that $A^+$ is complete to $K$ and so
$K\subseteq N(x)\cap N(y)$. We use this coloring to extend $c$ to the component
$K$. We claim that this does not create a monochromatic edge.  Indeed, if $uv$
is an edge where $c(u)=c(v)$, then $u,v$ are both neither in $K$ nor outside
$K$, since $c$ and $c_K$ are colorings.  Thus we may assume that $v\in K$ and
$u\in A\cup B\cup C$.  If $u\in A^+$, we have $c(u)\in\{1,2\}$, but
$c(v)\not\in\{1,2\}$ because $c_K$ does not use colors $1,2$. If $u\in A^-\cup
B\cup C$, then we have $L(u)=\{c(u)\}$ by \ref{clm:z1} and \ref{clm:z2}.
Therefore $c(u)$ was removed from $L(v)$ at Step 3.  But then $c(u)\neq c(v)$,
since $c(v)=c_K(v)\in L(v)$, a contradiction.

We repeat the above procedure for all connected component $K$ of $G[Y\cup Z]$ to
which it applies.  This properly extends $c$ to all these components. Note that
$c$ remains a good coloring throughout this process, since we never change any
colors of already colored vertices.  If there are no uncolored components, then
$c$ is a 4-coloring of $G$ and we are done. Thus we may assume that there still
remain uncolored components, and for every such component $K$, the vertices in
$A^+=S^+\cap N(K)$ use at least 3 distinct colors.

\begin{claim}\label{clm:z7}
Every uncolored connected component of $G[Y\cup Z]$ consists of a single vertex.
\end{claim}

\begin{proofclaim}
Let $K$ be an uncolored component of $G[Y\cup Z]$. Since $K$ is uncolored, then
(up to permuting colors) there are vertices $x_1,x_2,x_3\in A^+$ such that
$c(x_1)=1$, $c(x_2)=2$, and $c(x_3)=3$.  This implies $1\in L(x_1)$, $2\in
L(x_2)$, $3\in L(x_3)$ because $c$ is a good coloring.  By \ref{clm:z6}, the
three lists $L(x_1)$, $L(x_2)$, $L(x_3)$ are not parwise distinct.  So we may
assume by symmetry that $L(x_1)=L(x_2)$.  Thus $L(x_1)=L(x_2)=\{1,2\}$, since
$1\in L(x_1)$, $2\in L(x_2)$, and $|L(x_1)|=|L(x_2)|=2$ because $x_1,x_2\in
S^+$.  This implies $L(x_3)\neq \{1,2\}$, since $3\in L(x_3)$. So by
\ref{clm:z3}, we deduce that $x_1,x_2\in N(x_3)$. Now since $c(x_1)=1$ and
$c(x_2)=2$, and since $c$ is a good coloring, we have $(1,2)\not\in Q(x_1,x_2)$.
In particular, the pair $(1,2)$ was not added to $Q(x_1,x_2)$ at Step 5. Thus
there exists a 2-coloring $c_K$ of $G[K\cup \{x_3\}]$. For this note that
$K\cup\{x_3\}\subseteq N(x_1)\cap N(x_2)$. Now recall that $x_3$ is complete to
$K$, since $x_3\in A$.  Thus if $K$ contains an edge, then together with $x_3$
this forms a triangle in $G[K\cup \{x_3\}]$. This is impossible, since $c_K$ is
a 2-coloring of $G[K\cup\{x_3\}]$. Thus $K$ has no edge as claimed.

\end{proofclaim}

\begin{claim}\label{clm:z8}
Suppose that for all $u,v\in S^+$ the lists $L(u)$, $L(v)$ are either
identical or disjoint. Then $c$ is a nice coloring.
\end{claim}

\begin{proofclaim}
Suppose that for all $u,v\in S^+$, either $L(u)=L(v)$ or $L(u)\cap
L(v)=\emptyset$.  This implies (up to permuting colors) that for all $u\in S^+$,
the list $L(u)$ is either $\{1,2\}$ or $\{3,4\}$.  We prove that $c$ is a nice
coloring. Suppose otherwise. Then there exists an uncolored vertex $y$ such that
one of \ref{enum:n1}, \ref{enum:n2} fails.  By \ref{clm:z7}, every neighbor of
$y$ is colored, implying that \ref{enum:n1} holds. So \ref{enum:n2} must fail.
Let $U$ denote the set of all neighbors $u$ of $y$ with $c(u)\in\{1,2\}$. Let
$W$ denote the set of all neighbors $w$ of $y$ with $c(w)\in\{3,4\}$.

Since \ref{enum:n2} fails, each of $U$ and $W$ contains an adjacent pair of
vertices.  Let $u_1,u_2\in U$ be adjacent, and let $w_1,w_2\in W$ be adjacent.
Since $c$ is a coloring, since $u_1u_2\in E(G)$, and since $c(u)\in\{1,2\}$ for
each $u\in U$, we may assume $c(u_1)=1$ and $c(u_2)=2$. Similarly, $c(w_1)=3$
and $c(w_2)=4$.  Note that since $c$ is a good coloring, we have $1\in L(u_1)$,
$2\in L(u_2)$, $3\in L(w_1)$, and $4\in L(w_2)$.  Thus if all of
$u_1,u_2,w_1,w_2$ are in $A^+$, then $L(u_1)=L(u_2)=\{1,2\}$ and
$L(w_1)=L(w_2)=\{3,4\}$. This implies by \ref{clm:z3} that $\{u_1,u_2\}$ is
complete to $\{w_1,w_2\}$. But then $\{y,u_1,u_2,w_1,w_2\}$ forms a $K_5$ in
$G$, impossible since $G$ is clean.  Thus we may assume by symmetry that
$u_1\not\in A^+$. Therefore $u_1\in A^-\cup B\cup C$. In particular
$L(u_1)=\{c(u_1)\}=\{1\}$ by \ref{clm:z2}. Thus the color 1 was removed from
$L(u_2)$ at Step 3. This implies that $L(u_2)$ is neither $\{1,2\}$ nor
$\{3,4\}$, since $2\in L(u_2)$ but $1\not\in L(u_2)$. This shows that $u_2$ is
not in $A^+$. Thus $u_2\in A^-\cup B\cup C$ and so $L(u_2)=\{c(u_2)\}=\{2\}$ by
\ref{clm:z2}. It now follows that both color 1 and color 2 were removed from
$L(y)$ at Step 3, because $u_1,u_2\in N(y)$. Now, recall that $c(w_1)=3$ and
$c(w_2)=4$ and $c$ is a good coloring. Thus $(3,4)\not\in Q(w_1,w_2)$. In
particular, the pair $(3,4)$ was not added to $Q(w_1,w_2)$ at Step 5. Thus there
exists a 2-coloring $c_K$ of $G[N(w_1)\cap N(w_2)]$ using colors $\{1,2\}$ such
that $c_K(v)\in L(v)$ for all $v\in N(w_1)\cap N(w_2)$. However, $y\in
N(w_1)\cap N(w_2)$ and neither 1 nor 2 belongs to $L(y)$, a contradiction.
\end{proofclaim}

In view of \ref{clm:z8}, we may assume that there exist vertices in $S^+$ whose
lists are neither disjoint nor equal.  By \ref{clm:z6}, this implies (up to
permuting colors) that for every $u\in S^+$, the list $L(u)$ is either $\{1,2\}$
or $\{1,3\}$.  We show that in this case we can directly extend the coloring $c$
to the remaining vertices.  

\begin{claim}
Assign color 4 to every uncolored vertex. This yields a 4-coloring of $G$.
\end{claim}

\begin{proofclaim}
Suppose otherwise. Then there exists an uncolored vertex $y$ having a neighbor
$x_4$ with $c(x_4)=4$.  Since $\{y\}$ is a connected component of $G[Y\cup Z]$
by \ref{clm:z7}, and since $y$ is uncolored (i.e., neither of the first two
coloring steps applies), it follows that there are vertices $x_1,x_2,x_3\in A^+$
with pairwise distinct colors.  Since $c$ is a good coloring, we have $c(x_i)\in
L(x_i)$ for all $i\in\{1,2,3\}$.  Observe that $4\not\in L(u)$ for all $u\in
S^+$, since we assume that $L(u)$ is either $\{1,2\}$ or $\{1,3\}$ for any such
a vertex $u$.  Therefore, up to permuting colors, we have that $c(x_1)=1$,
$c(x_2)=2$, and $c(x_3)=3$.  Thus $1\in L(x_1)$, $2\in L(x_2)$, and $3\in
L(x_3)$ which implies $L(x_2)=\{1,2\}$ and $L(x_3)=\{1,3\}$. Moreover, up to
permuting colors (2 and 3), we may assume that $L(x_1)=\{1,2\}$.  Consequently,
by \ref{clm:z3}, we deduce that $x_1,x_2\in N(x_3)$.

Now recall that $c(x_4)=4$.  This implies $4\in L(x_4)$, since $c$ is a good
coloring.  Therefore $x_4\not\in S^+$, because $L(u)$ is one of $\{1,2\}$,
$\{1,3\}$ for all $u\in S^+$. This implies that $x_4\in A^-\cup B\cup C$. In
particular, $L(x_4)=\{c(x_4)\}=\{4\}$ by \ref{clm:z2}. Thus $4\not\in L(y)$ as
the color 4 was removed from $L(y)$ at Step 3, because $x_4\in N(y)$. Now recall
that $c(x_1)=1$ and $c(x_2)=2$. Since $c$ is a good coloring, $(1,2)\not\in
Q(x_1,x_2)$. Thus there exists a 2-coloring $c_K$ of $G[N(x_1)\cap N(x_2)]$ with
colors $\{3,4\}$ such that $c_K(v)\in L(v)$ for all $v\in N(x_1)\cap N(x_2)$.
Note that both $y$ and $x_3$ are in $N(x_1)\cap N(x_2)$.  Since $4\not\in L(y)$
but $c_K(y)\in L(y)$ and $c_K(y)\in\{3,4\}$, it follows that $c_K(y)=3$.
Similarly, $c(x_3)=3$, since $c_K(x_3)\in\{1,3\}$ and also $c_K(x_3)\in\{3,4\}$.
But this is impossible, since $x_3$ and $y$ are adjacent, yet they are given the
same color by $c_K$ which is a coloring, a contradiction.
\end{proofclaim}

This completes the proof.
\end{proof}

\subsection{Implementation and Complexity Analysis}

\begin{lemma}\label{lem:complexity}
Phase II algorithm has a polynomial-time implementation.
\end{lemma}

\begin{proof}
The implementation of Step 1 is straightfoward by directly testing all vertices
in $S$ whether or not to include them into $R$. This takes $O(n^2)$ time.  Note
that the set $R$ this step creates contains no more than 7 vertices, one for
each $i$.

Next, at Step 2, we choose a 4-coloring of $G[C\cup R]$. For each such a
coloring, we repeat Steps 2-6. Since $|R|\leq 7$ and every vertex in $R$ is
complete to adjacent vertices of $C$,  there are at most $7\times 2^6$ possible
4-colorings of $G[C\cup R]$ up to permuting colors. Thus the loop in Steps 2-6
repeats at most this many times.

In Step 3, we propagate the colors until the lists stabilize. Note that in this
process we only need to consider each vertex $v$ once. Indeed, once we remove
its only color from the lists of its neighbors, we do not need to do it again.
The whole process can be done in $O(n^2)$ by keeping a queue of candidate
vertices (with list of size 1), and adding a vertex into the queue once its list
is reduced to size 1. Note that some lists may become empty. This does not
affect the rest of the algorithm. Namely, if this happens, Step 6 will fail.

In Step 4, we consider each vertex $v$ and each color $i\in L(v)$.  We construct
auxiliary lists $L'$ and test if $G[N(v)]$ is list colorable with respect to
$L'$. To do this in polynomial time, we observe that the lists $L'$ do not
contain color $i$, by construction. So, in fact, the coloring we seek is a list
3-coloring using the remaining colors. For this, we use the algorithm of
\cite{3col-p6} which determines this in polynomial time.

In Step 5, we similarly consider auxiliary lists and construct a list coloring,
now using only two distinct colors. This has a straightforward reduction to an
$O(n^2)$ size instance of 2SAT \cite{edwards}.  Solving the instance takes
linear time in its size \cite{tarjan-sat}, thus $O(n^4)$ time altogether for
this step (including the construction).

Finally, at Step 6, we test if there is a 4-coloring of $G'=G[C\cup (X\setminus
Z)]$. If some list $L(v)$ is empty, we know right away that there is no such a
coloring. Otherwise, by \ref{clm:z2}, we deduce that the lists $L(v)$ of
vertices in $G'$ all have size 1 or 2.  Thus the test in Step 6 reduces to 2SAT
in a straightfoward manner as follows.

For every vertex $x$ and every color $i$, we have a variable $x_i$. The clauses
are as listed below.\smallskip

\begin{tabular}{cl}
$(x_i)$ & for each $x\in V(G')$ such that $L(x)=\{i\}$\smallskip\\
$(x_i\vee x_j)$ & for each $x\in V(G')$ such that $L(x)=\{i,j\}$\smallskip\\
$(\neg x_i\vee \neg x_j)$ & for each $x\in V(G')$ and for all distinct
$i,j\in\{1,2,3,4\}$\smallskip\\
$(\neg x_i\vee \neg y_i)$ & for all adjacent $x,y\in V(G')$ and each
$i\in\{1,2,3,4\}$\smallskip\\
$(\neg x_i)$ & for each $x\in V(G')$ and each $i\in Q(x)$\smallskip\\
$(\neg x_i\vee \neg y_j)$ & for all distinct $x,y\in V(G')$ and each
$(i,j)\in Q(x,y)$
\end{tabular}
\medskip

The first 4 types of clauses encode the property that $c$ is a 4-coloring
respecting the lists $L$. The last 2 types encode the conditions imposed by the
sets $Q(x)$ and $Q(x,y)$. It easily follows that this instance is satisfiable if
and only if the required coloring $c$ exists. In particular, from a satisfying
assignment, we construct a coloring $c$ by setting $c(x)=i$ if $x_i$ is true.
It is easy to see that the instance has $O(n^2)$ clauses and so again $O(n^2)$
time \cite{tarjan-sat} for this step (just like Step 5).

This shows that all steps have polynomial time complexity and they are repeated
a finite number of times. Altogether the complexity is polynomial as claimed.

This concludes the proof.
\end{proof}

We are now finally ready to prove Theorem \ref{thm:main}.\medskip

\section{Proof of Theorem \ref{thm:main}}

Let $G$ be a given $(P_6,C_5)$-free input graph.  As usual, let $n=|V(G)|$.

First, we apply Phase I algorithm to $G$.  If this algorithm declares that $G$
is not 4-colorable, then we stop and report this. If the algorithm declares that
$G$ is clean, then we proceed to the next stage. Otherwise, the algorithm
outputs a smaller graph $G'$ and we replace $G$ by $G'$. Then we apply Phase I
algorithm to $G$ again. We do this as long as $G$ is not clean. Since each
iteration makes the graph smaller, after at most $n$ iterations we must either
find that $G$ is not 4-colorable, or produce a clean graph $G$.  By Lemma
\ref{lem:cleaning}, this is correct.

In the next stage, we decompose $G$ by clique cutsets using the algorithm from
\cite{tarjan}. This produces in polynomial time a list $\cal G$ of subgraphs of
$G$ such that no graph in the list has a clique cutset and such that $G$ is
4-colorable if and only if every graph in the list is 4-colorable. The list
$\cal G$ contains at most $n$ graphs.

We then process every graph $G''\in \cal G$. First, we test if $G''$ is perfect.
To do this, it suffices to check for 7-antiholes (we do not need the full
generality of \cite{berge}).  If $G''$ is perfect, then we proceed to the next
graph in the list. Otherwise, $G''$ contains a 7-antihole and we apply Phase II
algorithm to $G''$. This determines in polynomial time if $G''$ is 4-colorable
as shown by Lemma \ref{lem:correct}.  If the algorithm finds that $G''$ is not
4-colorable, then we stop and report that $G$ is not 4-colorable.  Otherwise, we
proceed to the next graph in the list.

It we succesfully process all graphs in the list $\cal G$ without ever
rejecting, then we stop and declare that $G$ is 4-colorable.  The correctness of
this procedure follows from Lemma \ref{lem:cleaning} and Lemma
\ref{lem:correct}.  Since all steps have polynomial-time complexity, the overal
complexity of the algorithm is polynomial. 

Finally, if $G''\in\cal G$ is found to be 4-colorable, then a 4-coloring of
$G''$ can be found in polynomial time either by \cite{perfect-col} if $G''$ is
perfect, or by extending the coloring $c$ from Step 6 of Phase II algorithm if
$G''$ is not perfect, as described in the proof of Lemma \ref{lem:correct}.
Combining the colorings of all $G''\in{\cal G}$ and then reversing the
reductions performed by Phase I algorithm then finally produces a 4-coloring of
$G$. All these steps have a straightforward polynomial-time complexity.

This completes the proof.

\section*{Acknowledgement}

We are grateful to Rafael Saul Chudnovsky Panner without whose cooperation this
research would not have been possible.

\end{document}